\documentclass[sn-mathphys-num,Numbered]{sn-jnl}

\usepackage[T1]{fontenc}
\usepackage{graphicx}
\usepackage{algpseudocode}
\usepackage{algorithm}
\usepackage{amsmath}
\usepackage{amssymb}
\usepackage{multirow}
\usepackage{caption}
\usepackage[table]{xcolor}
\usepackage{amsthm}
\usepackage{booktabs}
\usepackage{longtable}
\usepackage{array}

\theoremstyle{definition}
\newtheorem{definition}{Definition}
\newtheorem{example}{Example}
\newtheorem{theorem}{Theorem}
\newtheorem{proposition}{Proposition}
\newtheorem{corollary}{Corollary}

\def\implies{\rightarrow}
\def\next{\mathcal{X}}
\def\until{~\mathcal{U}}
\def\entails{\vDash}
\def\doesnotentail{\nvDash}
\def\true{\top}
\def\false{\bot}
\def\maybe{\mathcal{M}}

\def\land{\wedge}

\begin{document}

\title[Dynamically Reprogrammable Monitors]{Dynamically Reprogrammable Runtime 
  Monitors for Bounded-time MTL}

\author[1]{\fnm{Chirantan} \sur{Hebballi}}\email{EE20BT013.alum24@iitdh.ac.in} 
\author[1]{\fnm{Akash} \sur{Poptani}}\email{EE20BT005.alum24@iitdh.ac.in}
\author[2]{\fnm{Amrutha} \sur{Benny}}\email{112004003@smail.iitpkd.ac.in}
\author*[1]{\fnm{Rajshekar} \sur{Kalayappan}}\email{rajshekar.k@iitdh.ac.in}
\author[2]{\fnm{Sandeep} \sur{Chandran}}\email{sandeepchandran@iitpkd.ac.in}
\author[1]{\fnm{Ramchandra} \sur{Phawade}}\email{prb@iitdh.ac.in}

\affil[1]{\orgname{Indian Institute of Technology Dharwad}, 
  \state{Karnataka}, 
  \postcode{580011}, 
  \country{India}
}
\affil[2]{\orgname{Indian Institute of Technology Palakkad}, 
  \state{Kerala}, 
  \postcode{678623}, 
  \country{India}
}

\abstract{
  A Runtime Verification (RV) framework that supports online, {\em at-speed}
verification of properties that can change dynamically (during in-field
operations) will benefit a large variety of applications. Several
state-of-the-art RV frameworks propose to implement monitors on FPGAs. While
this approach can support changes to the property being monitored during
in-field operations, they struggle to keep pace with the system under
verification which use high-performance processors. In this work, we propose a
novel, reprogrammable monitor that is implemented using standard cells instead
of FPGAs. This allows the monitor to be co-located with the system under
verification (on the same die), and hence is amenable to at-speed monitoring of
properties. Our proposed design consists of a programmable unit that implements
five basic operations and a set of queue-update rules. We show that a
composition of such programmable units faithfully implements discrete time,
bounded MTL. We demonstrate through simulations that our proposed monitor can be
reprogrammed (through its I/O pins) post deployment. A fairly large monitor
which can support MTL formulae upto $16$ atomic propositions occupies only $0.55 mm^2$,
while operating at a frequency of $1.25 GHz$.

}

\maketitle

\section{Introduction}
\label{sec:introduction}

Runtime Verification (RV) entails augmenting a System Under Verification (SUV)
with a monitor that continuously observes its functioning and flags any
deviation from the expected behavior. RV is particularly effective when the
monitor is non-intrusive and operates {\em at-speed} (online), in which case, the
monitor verifies the actions of the SUV in parallel and verdicts are issued
within a few clock cycles without affecting the timing (and functional) behavior
of the SUV. A monitor which allows the property to be verified to change over
time as SUV operates in-field is desirable when the SUV is a complex
safety-critical
system~\cite{r2u2,finkbeiner2009monitor,jakvsic2015signal}.

Several prior works have achieved this twin objective of at-speed verification
and flexibility by implementing the monitors on an
FPGA~\cite{lu2008automatic,r2u2,moosbrugger2017r2u2,pellizzoni2008hardware,reinbacher2011past}. 
An FPGA offers better performance while simultaneously achieving similar 
flexibility as compared to the corresponding software implementation of the 
monitor. This is possible because although an FPGA is a chip, it is built 
using Lookup Tables (LUTs), flip-flops (FFs), and multiplexers (MUXs). LUTs 
are simple storage elements (asynchronous SRAMs) that store all possible outputs 
of a Boolean function. Inputs (or arguments of the Boolean function) are used to
index into this table and fetch the appropriate output. New values are written 
into the LUTs when {\em reconfiguring} an FPGA, thereby allowing it to change its
functionality post deployment. Since the LUTs can be written as often as needed,
the synthesis procedures used to generate software monitors from input
specifications can be repurposed to generate different hardware monitors as
well.

There are three drawbacks with this approach. First,
an FPGA implementation is inefficient as compared to the same
circuit implemented using {\em standard cells}, as would be the case if they are
implemented within an Integrated Circuit (IC). A standard cell is an
implementation of basic gates such as NAND, which is optimized for a specific
technology node. For example, a two-input NAND cell in a standard cell library
has three-dimensional layout descriptions which a foundry can use to create and
arrange transistors that implement the NAND function. As a result, a design using
standard cells offers lower latency and higher throughput (operates at a higher clock
frequency -- capable of observing a larger number of events per second from the SUV),
consumes lesser power and occupies smaller area when implemented using standard 
cells instead of an FPGA. The second drawback is that reconfiguring an FPGA is 
time-consuming because it typically requires rewriting the contents of a few 
thousands to hundreds of thousands of LUTs. Finally, the bandwidth available for
communication between an FPGA and a processor core is limited even on 
state-of-the-art Systems-on-Chip (SoCs) such as Zynq Ultrascale+ that
integrate an FPGA and a processor on the same die~\cite{jindal2019dhoom}. A
flexible, reprogrammable monitor that is amenable to standard-cell
implementation do not have such restrictions, thereby enabling fine-grained
(every instruction) and at-speed monitoring of softwares running on processor
cores.

Several previous works have alleviated these drawbacks to
various extent. The drawback of increased reconfiguration time is addressed
through several techniques. One approach is to implement MTL operators as
fixed blocks (with a pre-defined number of instances of each operator being implemented), and route
the input stream through these blocks appropriately to monitor
properties~\cite{r2u2}. Another technique is to exploit sub-expressions that
are common across input properties to minimize the resources required, thereby
also reducing the amount of reconfiguration
required~\cite{kempa2020embedding}. Finally, it is also possible to use {\em
partial reconfiguration}, which state-of-the-art FPGAs support, where only a
small subset of LUTs are reconfigured instead of all the
LUTs~\cite{rahmatian2012adaptable}. The drawback of increased communication
latency between the processor and the FPGA fabric is alleviated by reducing the
volume of data to be transferred to the monitor by restricting the verification
to only bus transactions. This reduces the volume of data because bus
transactions are at least an order of magnitude fewer than the activity on the
register file (every instruction)~\cite{pellizzoni2008hardware}.
However, doing so requires the properties to be redefined over
a coarser (time) granularity.

In this work, we overcome all the three drawbacks of using FPGAs by synthesizing
hardware monitors that are amenable to be implemented using standard cells, 
and can be co-located with the processor, thereby enabling fine-grained and at-speed
monitoring. This necessitates a design approach that achieves flexibility
through dynamic reprogrammability instead of reconfiguration~\cite{farmltl}.
Under this approach, the circuit is fixed and input instructions
(byte-sequences) are used to configure the functioning of the monitor post
deployment. This approach has several benefits. The performance
parameters such as the frequency of the operation or the area consumed by the
monitoring circuit do not change with the input property. Moreover, since the
circuit design of the monitor is independent of the input property, and is
known at design-time, we can redesign the circuit (say pipeline it) to match
its operating frequency with that of the SUV, if required. This ensures that
the monitor keeps pace with the SUV (processes one input symbol every cycle).

We begin by designing an {\em Abstract Machine} (AM) that can be programmed to
perform one of five simple operations: \texttt{and}, \texttt{or}, \texttt{not},
\texttt{implies}, and \texttt{wire} (identity function), and update a custom
queue based on the result of the performed operation. We then prove that these
AMs can be programmed to realize Evaluator Machines (EMs) for {\em all} Metric
Temporal Logic (MTL)~\cite{alur1993real} operators. Finally, using programmable
interconnects, we compose several EMs together to realize a monitor for a given
MTL formula. Such composition of EMs is governed by the Abstract Syntax Tree
(AST) of the formula~\cite{r2u2,jakvsic2015signal}. The exact operation to be
performed at each AM, and the forwarding rules for the interconnect are
programmed before event streams are input to the AMs. Our proposed framework
consists of a parameterized hardware design, as well as a compiler that
generates programming bits to configure the AMs and the interconnects suitably.
Our experiments show that even a large MTL monitor (that supports upto 16 atomic
propositions and a maximum MTL operator time bound of 256 time steps) occupies
a modest area of $0.55 mm^2$ and operates at 1.25 $GHz$ clock frequency when
implemented using a conservative $32nm$ standard cell library.
This means that the developed monitor is capable of observing
$1.25 \times 10^9$ events from the SUV every second, which is significantly
higher than that possible with an FPGA-based monitor. Also, this high
efficiency -- high throughput and low area -- is achieved for every possible
MTL formula that can be expressed using 16 atomic propositions, 16 MTL
operators, and a maximum MTL operator time bound of 256 time steps. 
The proposed tool has been made available online~\cite{hebballi_2026_19084848}.

The rest of the paper is organized as follows. We discuss the background and
related works in Section~\ref{sec:related_main}. The design of the proposed
programmable MTL monitor is discussed in Section~\ref{sec:abstractmachines}.
Section~\ref{sec:arch} outlines an implementation that is amenable for
synthesis using standard cells, as well as the details for programming these
monitors. Section~\ref{sec:tool} discusses the developed 
tool~\cite{hebballi_2026_19084848} and the results of synthesizing MTL 
monitors of different size.  Finally, we conclude in Section~\ref{sec:conc}.

\section{Preliminaries and Related Work}
\label{sec:related_main}
First we describe syntax and semantics of Metric Temporal Logic. 
\subsection{Metric Temporal Logic (MTL)}
\label{sec:mtl}

An observation made about the SUV at any point in time is referred to as an {\em
  event}. An event consists of a set of valuations of Boolean
variables. The set of Boolean variables forms the set of atomic propositions
(APs) on which the desired properties are defined. Thus, we can formally state
that an event $e \in 2^{AP}$. The SUV is observed periodically and hence a run
or an execution of the SUV is an infinite sequence of events $e_0, e_1, e_2,
\ldots$, with $e_i$ referring to the observation made at time step $i$. In other
words, a run of the SUV can be represented as the function $\mathbb{N}_0
\rightarrow 2^{AP}$, where $\mathbb{N}_0$ is the set $\{0,1,2,\ldots\}$ 
of natural numbers. 

\textbf{Syntax}: If $\sigma$ is any atomic proposition, and if $\phi$ and $\psi$
are valid MTL formulae, then the following are valid MTL formulae as well:

\begin{center}
$true \mid \sigma \mid \neg \phi \mid \phi \lor \psi \mid 
\phi \land \psi \mid \phi \implies \psi \mid \next \phi \mid 
\phi \until_{[t_1, t_2]} \psi \mid \Box_{[t_1, t_2]} \phi \mid 
\diamond_{[t_1, t_2]} \phi$
\end{center}

Time bounds are specified as intervals: we write $[t_1, t_2]$ to
denote the set $\{ i \mid t_1 \leq i \leq t_2 ~\mbox{where}~ i, t_1, t_2 \in \mathbb{N}_0\}$.
We restrict ourselves to the bounded semantics, that is, $t_2$ cannot be infinity.
Discrete-time bounded semantics are well-suited to express many real-time
system specifications~\cite{reinbacher2014runtime,r2u2,jakvsic2015signal,monitoringMTL,bartocci2018specification}.\\

\noindent\textbf{Semantics}: We define the satisfaction of an MTL formula $\phi$ at time $i$
by an execution $x$, denoted $x^i \entails \phi$, inductively as follows:

\begin{itemize}
	\item $x^i \entails true$ is true,
	\item $x^i \entails \sigma$ iff $\sigma$ holds in $e_i$,
	\item $x^i \entails \neg\phi$ iff $x^i \doesnotentail \phi$,
	\item $x^i \entails \phi \lor \psi$ iff $x^i \entails \phi \lor x^i \entails \psi$,
	\item $x^i \entails \phi \land \psi$ iff $x^i \entails \phi \land x^i \entails \psi$,
	\item $x^i \entails \phi \implies \psi$ iff $x^i \doesnotentail \phi \lor x^i \entails \psi$,
	\item $x^i \entails \next\phi$ iff $x^{i+1} \entails \phi$,
	\item $x^i \entails \phi \until_{[t_1, t_2]} \psi$ iff $\exists j \in [i+t_1, i+t_2] : ( x^j \entails \psi \land (\forall k \in [i, j) :  x^k \entails \phi))$
	\item $x^i \entails \Box_{[t_1, t_2]} \phi$ iff $\forall j \in [i+t_1, i+t_2] : x^j \entails \phi$
	\item $x^i \entails \diamond_{[t_1, t_2]} \phi$ iff $\exists j \in [i+t_1, i+t_2] : x^j \entails \phi$
\end{itemize}

\begin{definition} (The Runtime Verification Problem). Given an MTL formula $\phi$ over
a set of atomic propositions $AP$, and a word $x$ over $2^{AP}$, the Runtime Verification problem
consists of constructing a word $x'$ over $\{true(\true),
false(\false)\}$, such that $x'(i)=\true$ iff
$x^i \entails \phi$.
\end{definition}

The word $x$ corresponds to the stream of observations from the SUV, and the
word $x'$ corresponds to the stream of verdicts.

For a given MTL formula $\phi$, its size $|\phi|$ 
is defined as the total number of appearances of
operators in it.  
The smallest number of time steps by which we can get the verdict for
$\phi$ is called its horizon, denoted by $H(\phi)$. It is recursively defined as follows:

\begin{itemize}
	\item $H(true) = H(\sigma) = 0$, $\forall\sigma \in AP$,
	\item $H(\neg\phi) =  1 + H(\phi)$,
	\item $H(\phi \lor \psi) = H(\phi \land \psi) = H(\phi \implies \psi) =  1 + max(H(\phi), H(\psi))$,
	\item $H(\next \phi) = 2 + H(\phi)$,
	\item $H(\phi \; \until_{[t_1, t_2]} \psi) =  t_2 + 1 + max(H(\phi), H(\psi))$
	\item $H(\Box_{[t_1, t_2]} \phi) = H(\diamond_{[t_1, t_2]} \phi) = t_2 + 1 + H(\phi)$
\end{itemize}
 
We define a map $l$ from MTL operators to $\mathbb{N}_0$: 
$l(\neg)=l(\land)= l(\lor)=l(\implies)=1; 
l(\next)=2;
l(\diamond_{[t_1, t_2]})= 
l(\Box_{[t_1, t_2]})= 
l(\until_{[t_1, t_2]})= t_2+1$.

Next we describe past approaches towards RV of MTL Specifications. 

\subsection{Related works}
\label{sec:related}

There are three broad approaches that prior works have used to construct MTL
monitors. The first approach is construct a table having $H(\phi)$ rows and
$|\phi|$ columns for a specification
$\phi$~\cite{maler2008checking,finkbeiner2009monitor,ho2014online,monitoringMTL,geilen2003improved}.
The content of a cell $(i, j)$ gives the verdict of subformula $j$ at time
$k-i$, as computed until time $k$. This verdict computed until time $k$ need not
be final verdict and may change as time progresses and more observations
about the SUV are available. At each time step, when a new event is received at
the monitor, the oldest row in the table is removed, and the table contents
shift by one row to make place for the newly received event. The input formula
is then solved in a bottom-up manner using the newly available AP valuations,
and the values of different cells in the table are updated. The sub-formulae in
the oldest row are guaranteed to have their final verdicts (follows from
$H(\phi)$). Several works have improved this further by: (i) not re-processing
values that were processed before, (ii) re-writing the formula to compute the
value of a cell to depend only on future values in cases where the cell value
cannot be computed using observations already at hand. This approach typically
requires $O(H \times |\phi|)$ space, and has a latency of $O(H \times |\phi|)$.
The throughput is typically $\frac{1}{H \times |\phi|}$, as for each new
observation, $O(H \times |\phi|)$ time is to be spent processing it. This
approach is popular among software monitors.

Another approach to construct MTL monitors is to construct the Abstract Syntax
Tree (AST) of the input formula and exploit the resulting tree structure to
derive high-throughput monitors that operate in a streaming manner (input values
are read only once; no
backtracking)~\cite{r2u2,jakvsic2015signal,selyunin2016monitoring,basin2011algorithms}.
These works instantiate an evaluator for each node in the tree (corresponds to
an MTL operator in the input specification). The results produced by an
evaluator module $c$ are conveyed to the evaluator module $p$ that represents the
parent of $c$ in the AST. At each time $i$, the event received at time $i$
enters the tree at the leaves, and propagates towards the root node in
subsequent time steps. The verdict corresponding to the event at time $i$ is
available at the root at time $i+O(H)$. The throughput is $\frac{1}{l_{max}}$,
where $l_{max}$ is the latency of the most complex evaluator module in the tree.
Our work also
adopts a similar approach where each node in the AST corresponds to a
programmable AM. Earlier works have used this approach to implement MTL monitors
on FPGAs~\cite{r2u2,jakvsic2015signal}. They have also stressed on the
importance of supporting changes to the input specifications in-field (or
programmability)~\cite{r2u2}. However, unlike our approach, they create a
separate evaluator module for each MTL operator, instantiate several evaluator
modules on the FPGA, and reconfigure the connections between these evaluators to
realize monitors for different MTL formulae. In contrast, our evaluator module
(or AM) is capable of evaluating all MTL operators and hence, we only need to
enable appropriate operation and/or select appropriate results produced to
realize monitors for different MTL formulae.

A third approach is to employ Timed Automata to capture MTL
specifications~\cite{maler2007synthesizing,divakaran2010conflict}. However,
there are no programmable Timed Automata Processors to the best of our
knowledge. Thus, programmability will have to be achieved through software or
FPGA implementations, which suffers from the drawbacks discussed in
Section~\ref{sec:introduction}.

Some works have explored the use of novel computing paradigms such as
processing-in-memory and neuromorphic computing to implement efficient
monitors~\cite{selyunin2016monitoring,wang2016overview}. These works are
orthogonal to our proposed work because commercially available System-on-Chips
(SoCs) are yet to integrate such accelerators on to the same die, which is
essential in our case to achieve at-speed monitoring.

Some researchers have studied runtime monitoring of quantative semantics of
MTL~\cite{chattopadhyay2020verified,donze2013signal,donze2013efficient}. This
involves not only determining the satisfaction of the given property, but also
determining the robustness of the SUV's execution, that is, how close the SUV
came to violating the property. In this work, we limit ourselves to the
qualitative semantics, as described in Section~\ref{sec:mtl}.

\section{Abstract Machines for MTL Formulae}
\label{sec:abstractmachines}

The building block of our proposed MTL monitor is an {\em Abstract Machine}
(AM), which is described in Section~\ref{sec:am_operators}.  The Runtime
Verification monitor of an MTL formula of size 1 is termed an Evaluator Machine
(EM). In Section~\ref{sec:realizing_evaluators}, we show how EMs for each MTL
operator can be realized by programming one or more AMs. We then show, in
Section~\ref{sec:em_composing}, how the RV monitor for an MTL formula of size
greater than 1 can be constructed by composing the EMs for each individual
operator in the formula, as directed by the AST of the formula.

\subsection{AM and its operations}
\label{sec:am_operators}
 
In each time step, an AM consumes one or two input operands, and produces one
output verdict. The input operands and the output belong to $\{true(\true),
false(\false)\}$.
 
An AM modifies a finite queue, which is used to hold partial results produced
when generating verdicts. We call this structure a {\em Que}.  Each cell of the
Que can store: (i) true ($\true$), (ii) false ($\false$), or (iii) {\em Maybe}
($\maybe$).  Let $\Gamma=\{\true, \false, \maybe\}$ denote the Que alphabet.
Let $Q^i$ denote the contents of the $Q$ at time instance $i$.  Let $Q^i[k]$
denote contents of $k$-th cell of $Q$ at time instance $i$ ($Q[k]$ denotes the
contents of $k$-th cell in the current time instance).  Let $Q[0]$ denote the
tail of $Q$.  Apart from the two basic operations of addition and deletion, $Q$
also supports a conditional $modify$ operation. These operations are defined
below.  Let $\emph{Intervals}$ denote the set of all intervals over
$\mathbb{N}_0$. 

The basic operations \emph{QOps} on a Que $Q$ are described below: 
\begin{itemize}  
	\item  $add(Q)$ pushes the contents of the que to higher index cells, and
		inserts element $\maybe$ i.e., $Q[0]=\maybe$. 
	\item  $del(Q, Head)$ deletes element $Q[Head]$ from the que. 
\item  $modify(Q,\mathcal{I},  \true_\maybe)$ conditionally modifies the que contents: for all $t$ in 
	interval $\mathcal{I}$, if $Q[t]=\maybe$ then $Q[t]=\true$
\item  $modify(Q,\mathcal{I},  \false_\maybe)$ 
 conditionally modifies the que contents: for all $t$ in interval  $\mathcal{I}$, if $Q[t]=\maybe$ then $Q[t]=\false$
\end{itemize}

An AM is a single state machine, so for the sake of simplicity we omit
mentioning state from the description. The inputs to the machine can be
classified into two categories: program inputs and operand inputs.

The program inputs to an AM are:
\noindent 
\begin{itemize}  
\item $opcode$ is an element of $OPCODES=\{not, and, or, implies, wire\}$  
\item $Q$ is a que, and $Head$ is an index of $Q$
\item $\mathcal{I}_\true$ and  $\mathcal{I}_\false$ are elements of $Intervals$
\item $Mod_{\true}$ and $Mod_{\false}$ are elements of $\{\true, \false\}$
\end{itemize}
The program inputs to an AM change only at the time of programming.
They stay unchanged in all other time steps.

The operand inputs $op_0, op_1$ are elements of $\{\true, \false\}$,
and they can change in every step.\\

The operation of an AM in each step is as follows. The machine first computes $res$:

\begin{equation*}
res=
    \begin{cases}
        \neg op_0 & \text{if } opcode = not \\
        op_0 \land op_1 & \text{if } opcode = and \\
        op_0 \lor op_1  & \text{if } opcode = or \\
        op_0 \implies op_1 & \text{if } opcode = implies \\
        op_0 & \text{if } opcode = wire 
    \end{cases}
\end{equation*}

It then performs a sequence of operations on $Q$:
\begin{itemize}  
\item  if $res = \true$ and $Mod_{\true}=\true$ then 
$add(Q); modify(Q, \mathcal{I}_\true, \true_\maybe); del(Q, Head)$; 
\item  if $res = \false$ and $Mod_{\false}=\true$ then 
$add(Q); modify(Q, \mathcal{I}_\false,
\false_\maybe);del(Q, Head)$; 
\item  if $res = \true$ and $Mod_{\true}=\false$ then $add(Q);del(Q, Head)$; 
\item  if $res = \false$ and $Mod_{\false}=\false$ then $add(Q);del(Q, Head)$; 
\end{itemize}

\subsection{Realization of Evaluator Machines for MTL operators}
\label{sec:realizing_evaluators}
 
Evaluator Machines (EMs) for MTL operators are also the monitors for MTL
formulae of size one. Let the formula be $\phi$ and  
let $\odot$ be  the single operator in it.  EMs have one or two input
streams and one output verdict stream.  At any given time step, the values of
input operands are $\alpha_0$ and $\alpha_1$.  An EM is composed of one or more
AM instantiations and a single Que.  Deleted values from the associated Que
form the stream of monitor verdicts.  At any time instant $i$, the contents of
the $k$-th que cell gives the verdict corresponding to time $(i-k)$. If $k \geq
l(\odot)$, then the verdict is stable (note that $l(\odot)=H(\phi)$ since $|\phi|=1$). 
Else, the verdict may change in future time steps. Thus, the value of $Head$ in the AM
instantiation is always chosen as greater than or equal to $l(\odot)$.

\begin{table}[htbp]
	\scriptsize
	\centering
	\caption{Realizing EMs corresponding to different MTL operators using AMs} \label{tab:em_operations}
	\begin{tabular}{|c||c|c|c|c|c|c|c|c|c|}
		\hline
		 & \multicolumn{9}{c|}{\textbf{AM Programming}} \\
		\cmidrule{2-10}
		\textbf{MTL} & \textbf{opcode} & \textbf{$op_0$} & \textbf{$op_1$} & \textbf{Que} & \textbf{Min. value} & \textbf{$\mathcal{I}_\true$} & \textbf{$\mathcal{I}_\false$} & \textbf{$Mod_\true$} & \textbf{$Mod_\false$} \\
		\textbf{operator} & & & & \textbf{ID} & \textbf{of Head} & & & & \\
		\hline
		\hline
		$\neg \alpha_0$ & not & $\alpha_0$ & - & $Q$ & $1$ & $[0, 0]$ & $[0, 0]$ & $\true$ & $\true$ \\
		\hline
		$\alpha_0 \lor \alpha_1$ & or & $\alpha_0$ & $\alpha_1$ & $Q$ & $1$ & $[0, 0]$ & $[0, 0]$ & $\true$ & $\true$ \\
		\hline
		$\alpha_0 \land \alpha_1$ & and & $\alpha_0$ & $\alpha_1$ & $Q$ & $1$ & $[0, 0]$ & $[0, 0]$ & $\true$ & $\true$ \\
		\hline
		$\alpha_0 \implies \alpha_1$ & implies & $\alpha_0$ & $\alpha_1$ & $Q$ & $1$ & $[0, 0]$ & $[0, 0]$ & $\true$ & $\true$ \\
		\hline
		$\next \alpha_0$ & wire & $\alpha_0$ & - & $Q$ & $2$ & $[1, 1]$ & $[1, 1]$ & $\true$ & $\true$ \\
		\hline
		$\Box_{[t_1, t_2]} \alpha_0$ & wire & $\alpha_0$ & - & $Q$ & $t_2+1$ & $[t_2, t_2]$ & $[t_1, t_2]$ & $\true$ & $\true$ \\
		\hline
		$\diamond_{[t_1, t_2]} \alpha_0$ & wire & $\alpha_0$ & - & $Q$ & $t_2+1$ & $[t_1, t_2]$ & $[t_2, t_2]$ & $\true$ & $\true$ \\
		\hline
		\multirow{3}{*}{$\alpha_0 \until_{[t_1, t_2]} \alpha_1$} & wire & $\alpha_0$ & - & $Q$ & $t_2+1$ & $[0,0]$ & $[0, t_1-1]$ & $\false$ & $\true$ \\
		 & wire & $\alpha_1$ & - & $Q$ & $t_2+1$ & $[t_1, t_2]$ & $[t_2, t_2]$ & $\true$ & $\true$ \\
		 & or & $\alpha_0$ & $\alpha_1$ & $Q$ & $t_2+1$ & $[0,0]$ & $[t_1, t_2-1]$ & $\false$ & $\true$ \\
		\hline
		\multirow{2}{*}{$\alpha_0 \until_{[0, t_2]} \alpha_1$} & or & $\alpha_0$ & $\alpha_1$ & $Q$ & $t_2+1$ & $[0,0]$ & $[0, t_2-1]$ & $\false$ & $\true$ \\
		 & wire & $\alpha_1$ & - & $Q$ & $t_2+1$ & $[0, t_2]$ & $[t_2, t_2]$ & $\true$ & $\true$ \\
		\hline
	\end{tabular}
\end{table}

Table~\ref{tab:em_operations} lists how EMs for different MTL operators can be
realized through AM instantiations.  We design the EM in such a way that it
processes any given input operand only once, enabling us to build a perfectly
streaming monitor. The effect that the input operands in the current time step
have on verdicts corresponding to past time steps is recorded in the result
queue appropriately.

\begin{table}
\vspace{-0.65cm}
\centering
\caption{Example illustrating working of EM for ($ \neg \alpha_0$) \label{tab:example_neg}}
\scriptsize
\bgroup
\def\arraystretch{1.2}
\setlength\tabcolsep{0.1cm}
\begin{tabular}{|c|c|c|c|l|l|c|c|c|c|c|}
  \hline
	\multirow{2}{*}{time} & \multirow{2}{*}{$\alpha_0$} & \multirow{2}{*}{$\alpha_1$} & \multirow{2}{*}{res} & \multicolumn{1}{c|}{Que modification} & \multicolumn{5}{c|}{Que contents} & \multirow{2}{*}{$verdict$} \\
  \cmidrule{6-10}
	& & & & [$<$interval$>$]: $<$value$>$ & & 0 & 1 & 2 & 3 & \\
  \hline
	t = 0 & $\true$  & - & $\false$ &                             & after $add$    & $\maybe$ & & & & \\
	      &          &   &          & [0,0]$:\false_\maybe$       & after $modify$ & $\false$ & & & & \\
	      &          &   &          &                             & after $del$    & $\false$ & & & & \\
  \hline
	t = 1 & $\false$ & - & $\true$  &                             & after $add$    & $\maybe$ & $\false$ & & & \\
	      &          &   &          & [0,0]$:\true_\maybe$        & after $modify$ & $\true$  & $\false$ & & & \\
	      &          &   &          &                             & after $del$    & $\true$  &          & & & $r^0 = \false$ \\
  \hline
	t = 2 & $\false$ & - & $\true$  &                             & after $add$    & $\maybe$ & $\true$  & & & \\
	      &          &   &          & [0,0]$:\true_\maybe$        & after $modify$ & $\true$  & $\true$  & & & \\
	      &          &   &          &                             & after $del$    & $\true$  &          & & & $r^1 = \true$ \\
  \hline
\end{tabular}
\egroup
\end{table}

We demonstrate the realization and working of EMs through two examples.

\begin{example}
$EM$ for $\neg$ is an AM $M=(opcode=\neg, Q, Head=1, \mathcal{I}_\true=[0,0],
	\mathcal{I}_\false=[0,0], Mod_\true=\true, Mod_\false=\true$), as given
	in Table~\ref{tab:em_operations}.  The operation of this machine
	follows from the definition of an AM, and is illustrated in
	Table~\ref{tab:example_neg}. 

	When $\alpha_0 = \true$, we have $res = \false$ (see row for $time=0$
	in Table~\ref{tab:example_neg}).  Since we have $Mod_\false=\true$, we
	do the following  operations on $Q$ : $ add (Q,\maybe);$ $modify(Q,
	[0,0], \false_\maybe);$ $del(Q,1);$.  Since there is nothing at
	position $1$, the $del$ operation does not do anything.

At time=$1$, $\alpha_0 = \false$, and so we have $res = \true$.  Since we have
	$Mod_\true=\true$, we do the following  operations on $Q$ : $ add
	(Q,\maybe);$ $modify(Q, [0,0], \true_\maybe);$ $del(Q,1);$.  The
	deleted value forms the verdict corresponding to time
	$(current\_time-Head)=(1-1)=0$, that is, $r^0 = \false$ (we denote the
	verdict corresponding to time $i$ by $r^i$). 
	
	Similarly, at time=$2$, the verdict corresponding to time $1$
	($r^1=\true$) is output by the EM. We thus get a new verdict every
	time step.
\end{example}

\begin{table}
\centering
\caption{Example illustrating working of EM for ($\alpha_0 \until_{[1,2]} \alpha_1$) \label{tab:example_until}}
\scriptsize
\bgroup
\def\arraystretch{1.2}
\setlength\tabcolsep{0.1cm}
\begin{tabular}{|c|c|c|c|l|l|c|c|c|c|c|}
  \hline
	\multirow{2}{*}{time} & \multirow{2}{*}{$\alpha_0$} & \multirow{2}{*}{$\alpha_1$} & \multirow{2}{*}{res} & \multicolumn{1}{c|}{Que modification} & \multicolumn{5}{c|}{Que contents} & \multirow{2}{*}{$verdict$} \\
  \cmidrule{6-10}
	& & & & [$<$interval$>$]: $<$value$>$ & & 0 & 1 & 2 & 3 & \\
  \hline
	t = 0 & $\false$ & $\false$ & $\false$, $\false$, $\false$ &                                                                                        & after $add$    & $\maybe$ & & & & \\
	      &          &          &                              & [0,0]$:\false_\maybe$, [2,2]$:\false_\maybe$, [1,1]$:\false_\maybe$                    & after $modify$ & $\false$ & & & & \\
	      &          &          &                              &                                                                                        & after $del$    & $\false$ & & & & \\
  \hline
	t = 1 & $\true$  & $\false$ & $\true$, $\false$, $\true$   &                                                                                        & after $add$    & $\maybe$ & $\false$ & & & \\
	      &          &          &                              & [2,2]$:\false_\maybe$                                                                  & after $modify$ & $\maybe$ & $\false$ & & & \\
	      &          &          &                              &                                                                                        & after $del$    & $\maybe$ & $\false$ & & & \\
  \hline
	t = 2 & $\true$  & $\false$ & $\true$, $\false$, $\true$   &                                                                                        & after $add$    & $\maybe$ & $\maybe$ & $\false$ & & \\
	      &          &          &                              & [2,2]$:\false_\maybe$                                                                  & after $modify$ & $\maybe$ & $\maybe$ & $\false$ & & \\
	      &          &          &                              &                                                                                        & after $del$    & $\maybe$ & $\maybe$ & $\false$ & & \\
  \hline
	t = 3 & $\false$ & $\true$  & $\false$, $\true$, $\true$   &                                                                                        & after $add$    & $\maybe$ & $\maybe$ & $\maybe$ & $\false$ & \\
	      &          &          &                              & [0,0]$:\false_\maybe$, [1,2]$:\true_\maybe$                                            & after $modify$ & $\false$ & $\true$  & $\true$  & $\false$ & \\
	      &          &          &                              &                                                                                        & after $del$    & $\false$ & $\true$  & $\true$  &          & $r^0=\false$ \\
  \hline
	t = 4 & $\true$  & $\true$  & $\true$, $\true$, $\true$    &                                                                                        & after $add$    & $\maybe$ & $\false$ & $\true$  & $\true$ & \\
	      &          &          &                              & [1,2]$:\true_\maybe$                                                                   & after $modify$ & $\maybe$ & $\false$ & $\true$  & $\true$ & \\
	      &          &          &                              &                                                                                        & after $del$    & $\maybe$ & $\false$ & $\true$  &          & $r^1=\true$ \\
  \hline
\end{tabular}
\egroup
\end{table}

\begin{example}
	$EM$ for $\until_{[1,2]}$ is a set of three AMs all operating on the
	same Que $Q$ as given in Table~\ref{tab:em_operations}.  The operation
	of this machine follows from the definition of AM, and is illustrated
	in Table~\ref{tab:example_until}.

At each time step, each of the three AMs perform the given operation ($(wire\;
	\alpha_0)$, $(wire\; \alpha_1)$, and $(\alpha_0\; or\; \alpha_1)$
	respectively).  The three results obtained are indicated in the $res$
	column in Table~\ref{tab:example_until}. In time step $0$, the first AM
	computes $res=\false$ and since its $Mod_\false=\true$, it modifies 
	$Q$ cells with indices in the range $[0, (t_1-1)]=[0,0]$ to $\false$,
	since the cell's old value was $\maybe$. Similarly, the second AM
	attempts to modify cell $2$ to $\false$ (if $\maybe$), and the third AM
	attempts to modify cell $1$ to $\false$ (if $\maybe$). But since these
	cells are empty, there is no change to $Q$.

	In time step $1$, the first and third AMs do not modify $Q$ since their
	$res$ values are $\true$, and their $Mod_\true$ values are $\false$.
	However, the second AM computes $res=\false$, and since its
	$Mod_\false=\true$, it attempts to modify the $Q$ cell $t_2=2$ to
	$\false$ (if $\maybe$). However, since this cell is empty, $Q$ remains
	unmodified.  The updates to $Q$ progress in a similar fashion in
	subsequent time steps.

	In time step $3$, the verdict corresponding to time
	$(current\_time-Head)$$=0$, that is, $r^0 = \false$ is output by
	the EM. Similarly, in time step $4$, the verdict corresponding to time
	$1$ ($r^1=\true$) is output. 
\end{example}

\subsubsection{Correctness of Evaluator Machines}
The expectation is that at any time $i$, the verdict for single-operator
formula $\phi$ corresponding to time $i-k$ is present in cell $Q^i[k]$, for all
$l(\odot) \leq k \leq i$ (assuming $k < Head$).  We describe this correctness
requirement through two theorems. Theorem~\ref{lab-correctness-thm} captures
that an Evaluator Machine correctly evaluates the verdict of $\phi$
corresponding to time $i-l(\odot)$, in cell $Q^i[l(\odot)]$, at time $i$.
Theorem~\ref{lab-correctness-thm2} captures that once a verdict is given, it is
not altered as long as it is present in the $Q$. Proofs of
Theorem~\ref{lab-correctness-thm} and  Theorem~\ref{lab-correctness-thm2}, are
given in the Appendix.

\begin{theorem}       
Let $\phi$ be a single-operator MTL formula and $\odot$ be the operator in it. 
Then in an EM for $\phi$, for all time instances $i$, with $i \geq l(\odot)$,  
\begin{center} 
	{\em if} $x^{i-l(\odot)} \entails \phi${\em, then} $Q^{i}[l(\odot)]=\true$.  
\end{center} 
\begin{center} 
  {\em if} $x^{i-l(\odot)} \doesnotentail \phi${\em, then} $Q^{i}[l(\odot)]=\false$.  
\end{center} 
\label{lab-correctness-thm} 
\end{theorem}

\begin{theorem}
Let $\phi$ be a single-operator MTL formula and $\odot$ be the operator in it. 
Then in an EM for $\phi$,
	for all $j$ such that $j \geq l(\odot)$,
	and for all $k$ such that $1 \leq k \leq (Head-l(\odot))$,
	$Q^{(j+k)}[l(\odot) + k] = Q^j[l(\odot)]$.
\label{lab-correctness-thm2}
\end{theorem}

Since we use $Head \geq l(\odot)$, by Theorem~\ref{lab-correctness-thm} and 
Theorem~\ref{lab-correctness-thm2}, we get the following corollary establishing 
correctness of Evaluator Machines.  
\begin{corollary}
For a given MTL formula $\phi$, and an input word $x$, 
the stream of deleted values from the Que gives the stream of verdicts
$x'$ such that $x'(i)=\true$ iff $x^i \entails \phi$.
\end{corollary}

\begin{proposition}
	\label{prop:simul_modify}
	In the EM for $\until_{[t_1,t_2]}$ operator, the three AMs
	that modify a single que $Q$, do not simultaneously (in the same time step) modify the same cells of $Q$.
\end{proposition}
The same holds for the $\until_{[0,t_2]}$ operator as well. The proof is given in the Appendix. \\
 
It bears noting that all AMs are homogeneous, and the EM for
any MTL operator can be realized using any of the AMs. This high flexibility is
in contrast to the typical FPGA approaches that implement a pre-defined number
of instances of each MTL operator, and program the connectivity between them at
runtime to change the property being monitored. This FPGA approach limits the
number of MTL formulae that can be captured.

We describe realization of monitors for MTL formulae of size greater
than $1$ next.  

\subsection{Monitors for MTL Formulae}
\label{sec:em_composing}

To evaluate any given MTL formula, we construct its abstract syntax tree (AST).
For each operator in the formula, a different EM is given the responsibility
for its evaluation.  The outputs of some EMs are connected to the inputs of
other EMs according to the AST.  If an operator $x$ is a child of an operator
$y$ in the AST, then the EM ($EM_y$) corresponding to $y$, must read its source
operand from the queue of the EM ($EM_x$) corresponding to $x$.  Additionally,
$EM_y$ must read from the correct position in the queue  of $EM_x$ such that
functional correctness is maintained.  Which means the $Head$ of $EM_x$ needs
to be set correctly.  
 
Please note that functions $l$ and $H$ are the same for single
operators.  By Theorem~\ref{lab-correctness-thm}, the value of $Head$
for every $EM_x$ should be at least $l(x)$. However, there are
scenarios where $Head_x$ must be greater than $l(x)$ to maintain
functional correctness, as illustrated in
Example~\ref{example:balancing}.

\begin{figure}[h!]
	\centering
	\begin{tabular}{c}
		\begin{minipage}{0.95\textwidth}
			\centering
			\includegraphics[width=0.9\textwidth]{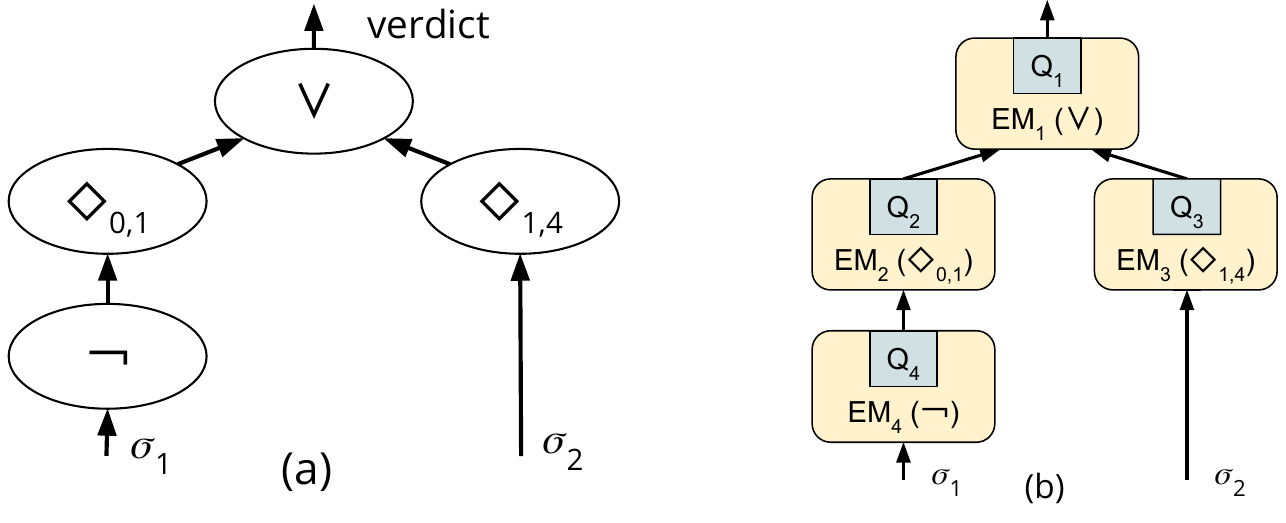}
			\caption{AST and EMs for $\diamond_{[0, 1]}\neg\sigma_1 \lor \diamond_{[1,4]}\sigma_2$} \label{fig:am_example}
		\end{minipage}
		\\
		\begin{minipage}{0.95\textwidth}
			\centering
			\begin{tabular}{|c|c|c|c|c|c|c|c|c||c|c|c|c|c|c|c|c|c|}
				\hline
				\multicolumn{9}{|c||}{\textbf{Scenario 1: $Head_2=2$}} & \multicolumn{9}{c|}{\textbf{Scenario 2: $Head_2=3$}}\\
				\hline
				time  & 0 & 1 & 2 & 3 & 4 & 5 & 6 & 7 &
          time & 0 & 1 & 2 & 3 & 4 & 5 & 6 & 7 \\
				\hline
				$Q_4[0]$ & $r_4^0$ & $r_4^1$ & $r_4^2$ & $r_4^3$ & $r_4^4$ & $r_4^5$ & $r_4^6$ & $r_4^7$ &
           $Q_4[0]$ & $r_4^0$ & $r_4^1$ & $r_4^2$ & $r_4^3$ & $r_4^4$ & $r_4^5$ & $r_4^6$ & $r_4^7$ \\
				\hline
				$Q_4[1]$ & & $r_4^0$ & $r_4^1$ & $r_4^2$ & $r_4^3$ & $r_4^4$ & $r_4^5$ & $r_4^6$ &
				$Q_4[1]$ & & $r_4^0$ & $r_4^1$ & $r_4^2$ & $r_4^3$ & $r_4^4$ & $r_4^5$ & $r_4^6$ \\
				\hline
				$Q_2[0]$ & &         & $r_2^0$ & $r_2^1$ & $r_2^2$ & $r_2^3$ & $r_2^4$ & $r_2^5$ &
				$Q_2[0]$ & &         & $r_2^0$ & $r_2^1$ & $r_2^2$ & $r_2^3$ & $r_2^4$ & $r_2^5$ \\
				\hline
				$Q_2[1]$ & &         &         & $r_2^0$ & $r_2^1$ & $r_2^2$ & $r_2^3$ & $r_2^4$ &
                                $Q_2[1]$ & &         &         & $r_2^0$ & $r_2^1$ & $r_2^2$ & $r_2^3$ & $r_2^4$ \\
				\hline
				$Q_2[2]$ & &         &         &         & $r_2^0$ & \cellcolor{red!40}$r_2^1$ & \cellcolor{red!25}$r_2^2$ & \cellcolor{red!25}$r_2^3$ &
                                $Q_2[2]$ & &         &         &         & $r_2^0$ & $r_2^1$ & $r_2^2$ & $r_2^3$\\
				\hline
          & & & & & & & & &
				$Q_2[3]$ & &         &         &         &         & \cellcolor{green!40}$r_2^0$ & \cellcolor{green!25}$r_2^1$ & \cellcolor{green!25}$r_2^2$ \\
				\hline
				$Q_3[0]$ & $r_3^0$ & $r_3^1$ & $r_3^2$ & $r_3^3$ & $r_3^4$ & $r_3^5$ & $r_3^6$ & $r_3^7$ &
           $Q_3[0]$ & $r_3^0$ & $r_3^1$ & $r_3^2$ & $r_3^3$ & $r_3^4$ & $r_3^5$ & $r_3^6$ & $r_3^7$ \\
				\hline
				$Q_3[1]$ & & $r_3^0$ & $r_3^1$ & $r_3^2$ & $r_3^3$ & $r_3^4$ & $r_3^5$ & $r_3^6$ &
           $Q_3[2]$ & & $r_3^0$ & $r_3^1$ & $r_3^2$ & $r_3^3$ & $r_3^4$ & $r_3^5$ & $r_3^6$ \\
				\hline
				$Q_3[2]$ & & & $r_3^0$ & $r_3^1$ & $r_3^2$ & $r_3^3$ & $r_3^4$ & $r_3^5$ &
           $Q_3[2]$ & & & $r_3^0$ & $r_3^1$ & $r_3^2$ & $r_3^3$ & $r_3^4$ & $r_3^5$ \\
				\hline
				$Q_3[3]$ & & & & $r_3^0$ & $r_3^1$ & $r_3^2$ & $r_3^3$ & $r_3^4$ &
           $Q_3[3]$ & & & & $r_3^0$ & $r_3^1$ & $r_3^2$ & $r_3^3$ & $r_3^4$ \\
				\hline
				$Q_3[4]$ & & & & & $r_3^0$ & $r_3^1$ & $r_3^2$ & $r_3^3$ &
           $Q_3[4]$ & & & & & $r_3^0$ & $r_3^1$ & $r_3^2$ & $r_3^3$ \\
				\hline
				$Q_3[5]$ & & & & & & \cellcolor{red!40}$r_3^0$ & \cellcolor{red!25}$r_3^1$ & \cellcolor{red!25}$r_3^2$ &
                                $Q_3[5]$ & & & & & & \cellcolor{green!40}$r_3^0$ & \cellcolor{green!25}$r_3^1$ & \cellcolor{green!25}$r_3^2$\\
				\hline
				$Q_1[0]$ & & & & & & & \cellcolor{red!40}$r_1^0$ & \cellcolor{red!25}$r_1^1$ &
                                $Q_1[0]$ & & & & & & & \cellcolor{green!40}$r_1^0$ & \cellcolor{green!25}$r_1^1$ \\
				\hline
				$Q_1[1]$ & & & & & & & & \cellcolor{red!40}$r_1^0$ &
				$Q_1[1]$ & & & & & & & & \cellcolor{green!40}$r_1^0$ \\
				\hline
				\hline
				\multicolumn{18}{|c|}{$r_i^j$: verdict
produced by $EM_i$ corresponding to time $j$} \\
				\hline
			\end{tabular}
			\captionof{table}{Example motivating the importance of $Head$ value choice} \label{tab:am_example_timeline}
		\end{minipage}
		\\
	\end{tabular}
\end{figure}

\begin{example}
\label{example:balancing}
Consider the formula $\diamond_{[0, 1]}\neg\sigma_1 \lor \diamond_{[1, 4]}\sigma_2$. 
The AST corresponding to this formula is given in Figure~\ref{fig:am_example}(a) and 
the allocation of operators to EMs is given in Figure~\ref{fig:am_example}(b). 
 
Table~\ref{tab:am_example_timeline} shows the values in different queue cells
	in two scenarios. $r^j_i$ denotes the verdict computed by $EM_i$
	corresponding to time $j$. In the first scenario, $Head$ of every $EM$
	is set to the $l()$ value of the corresponding operator. This gives
	$Head_1=1$, $Head_{2}=2$, $Head_{3}=5$, and $Head_{4}=1$.  In time step
	$2$, $EM_2$ takes in the verdict $r^0_4$ from cell $Q^1_4[1]$ and
	begins its computation of verdict $r^0_2$ in cell $Q^2_2[0]$.
	Likewise, in time step $6$, $EM_1$ performs the $\lor$ operation on the
	contents of cells $Q_2^5[2]$ and $Q_3^5[5]$. However, these
	respectively correspond to the verdict of $EM_2$ corresponding to time
	$1$ ($r_2^1$) and the verdict of $EM_3$ corresponding to time $0$
	($r_3^0$). This difference in time is functionally incorrect, and a
	similar mismatch in the time of the source operands of $EM_1$ is seen
	in every subsequent cycle.

	We delve deeper to find the cause of this functional incorrectness. As
	seen in the illustration in Table~\ref{tab:am_example_timeline}, $EM_4$
	takes $2$ time steps to produce a verdict (verdict passes through
	$Q_4[0]$ and $Q_4[1]$).  Likewise, $EM_2$ takes $3$ time steps and
	$EM_3$ takes $6$ time steps. This means that the verdict of $EM_1$'s
	left-subtree corresponding to time $i$ is available at $EM_1$ at time
	$i+2+3=i+5$, while the right-subtree verdict is available at time
	$i+6$. This mismatch causes the functional incorrectness.

	We can correct this by setting $Head_{2}=3$, as done in the second
	scenario in Table~\ref{tab:am_example_timeline}. Now in time $6$,
	$Q_2^5[3]$ contains $r_2^0$ (unchanged from $Q_2^4[2]$, as per
	Theorem~\ref{lab-correctness-thm2}), thereby resulting in functional
	correctness. We see that in every cycle, $EM_1$ receives operands
	corresponding to the same time as required.
\end{example}

    \begin{algorithm}[hbt]
        \caption{Computing $Head$s of EMs in the monitor}
        \label{alg:headcompute}
	\scriptsize
        \begin{algorithmic}[1] 
            \Function{ComputeHead}{$EM_{cur}$ which evaluates MTL operator $\sigma$}
                \State $Head_{cur} \gets l(\sigma)$
                \If{$EM_{cur}$ is leaf node}
		   \State $Height_{cur}=Head_{cur}+1$

                \ElsIf{$\sigma$ is a unary operator}
                   \State $(Head_{chld}, Height_{chld}) \gets \Call{ComputeHead}{Child~of~EM}$
		   \State $Height_{cur} \gets Head_{cur} + 1 + Height_{chld}$

                \ElsIf{$\sigma$ is a binary operator}
                   \State $(Head_{lchld}, Height_{lchld}) \gets \Call{ComputeHead}{Left~Child~of~EM}$
                   \State $(Head_{rchld}, Height_{rchld}) \gets
\Call{ComputeHead}{Right~Child~of~EM}$
	                  \If {$Height_{lchld}=Height_{rchld}$} 
				\State $Height_{cur} \gets Head_{cur} + 1 +
Height_{lchld}$
 
	                  \ElsIf {$Height_{lchld} < Height_{rchld}$} 
				\State $Height_{cur} \gets Head_{cur} + 1 + Height_{rchld}$
				\State $Head_{lchld} \gets
Head_{lchld} + Height_{rchld} - Height_{lchld}$
                	   	\State {$Height_{lchld} \gets Height_{rchld}$} 
 
 	                  \ElsIf {$Height_{rchld} < Height_{lchld}$} 
				\State $Height_{cur} \gets Head_{cur} + 1 + Height_{lchld}$
				\State $Head_{rchld} \gets
Head_{rchld} + Height_{lchld} - Height_{rchld}$
                	   	\State {$Height_{rchld} \gets Height_{lchld}$} 
			  \EndIf
 
                \EndIf

                \State \Return ($Head_{cur}, Height_{cur}$)

            \EndFunction
        \end{algorithmic}
    \end{algorithm}

Algorithm~\ref{alg:headcompute} gives a recursive procedure to be followed to
compute the correct value of $Head$ for an EM. Given a formula,
to set the $Head$s of all the EMs in its monitor correctly, we call
this function on the root node of the formula's AST. This computes the $Head$s of the nodes
starting from the leaves and moves towards the root, balancing the
heights of the subtrees whenever a binary operator is encountered.

\section{Programmable MTL Monitor}
\label{sec:arch}

\begin{figure}[htbp]
		\centering
		\includegraphics[width=0.95\textwidth]{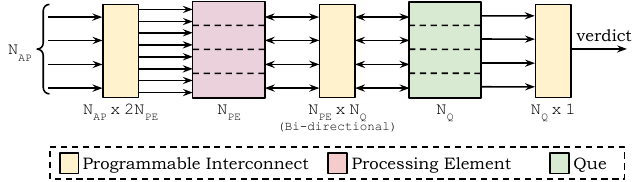}
		\caption{Architecture of the proposed MTL monitor (for $N=4$)} \label{fig:architecture}
\end{figure}

Figure~\ref{fig:architecture} shows the hardware design of the programmable MTL
monitor. The overall architecture can be broadly viewed as consisting of three
components: (i) Programmable Interconnects, (ii) $N_{PE}$ number of {\em
Processing Elements (PEs)}, and (iii) $N_{Q}$ number of {\em Que elements
(Qs)}.  A PE is an implementation of  the AM and a Que is an implementation of
the Que discussed in Section~\ref{sec:am_operators}.  The monitor supports upto
$N_{AP}$ number of APs.  The APs are connected to the PEs through a
programmable interconnect which allows communication in the {\em AP2PE} (AP
$\implies$ PE) direction.  The input AP values can be forwarded to any of the
inputs of a PE. Similarly, the PEs and Qs are connected through another
programmable interconnect which supports communication from any of the PEs to
any of the Qs, and vice versa.  Therefore, it supports upto $N_Q$ simultaneous
connections in the {\em PE2Q} (PE $\implies$ Q) direction, and upto $2N_{PE}$
connections in the {\em Q2PE} (Q $\implies$ PE) direction.  Each Q internally
implements a buffer of size $Q_{SZ}$.  A larger $Q_{SZ}$ supports a larger
temporal interval for an MTL operator.  To realize an EM for an MTL operator,
we program one or more PEs, as well as the Interconnect to connect the chosen
PEs to an appropriate Que.  A monitor for MTL formula $\phi$ is realized by
programming the Interconnects to connect the different Ques to the appropriate
PEs, according to the AST of $\phi$.  The generated verdict for $\phi$ is
available in the Que of the EM which corresponds to AST root node.  This
verdict is forwarded to the output pin through another programmable
interconnect (carrying traffic in the {\em Q2OUT} direction). This arrangement
allows any available PE and Q to be used to implement an MTL operator.

\subsection{Processing Element}
\label{sec:pe}

\begin{table}[htbp]
	\scriptsize
	\centering
	\caption{Programming of a Processing Element} \label{tab:pe_program}
  \setlength{\tabcolsep}{4pt} 
  \renewcommand{\arraystretch}{1.35}
	\begin{tabular}{|c|p{8cm}|c|}
		\hline
		\textbf{Field} & \multicolumn{1}{c|}{\textbf{Description}} & \textbf{\# bits} \\
		\hline
		\hline
		{\em isActive} & Is this PE being used & 1 \\
		\hline
		{\em op0Src} & First operand's source ({\tt 0:AP}, {\tt 1:Q}) & 1 \\
		\hline
		{\em op1Src} & Second operand's source ({\tt 0:AP}, {\tt 1:Q}) & 1 \\
		\hline
		{\em opcode} & Opcode ({\tt 000:wire}, {\tt 001:not}, {\tt 010:or}, {\tt 011:and}, and {\tt 100:implies}) & 3 \\
		\hline
		{\em r\_qid} & The ID of the queue where the results of this PE are placed & $\lceil \log_2N_Q \rceil$ \\
		\hline
		$\mathcal{I}_{\true}$ & Corresponds to the $\mathcal{I}_{\true}$  in the AM definition & $2 \lceil \log_2{Q_{SZ}} \rceil$ \\ 
		\hline
		$\mathcal{I}_{\false}$ & Corresponds to the $\mathcal{I}_{\false}$  in the AM definition & $2 \lceil \log_2{Q_{SZ}} \rceil$ \\ 
		\hline
	\end{tabular}
\end{table}

A PE can be programmed using instructions of the format described
in~Table~\ref{tab:pe_program}. The entries in Field column follow  the
definition of an AM. 
A PE may receive one or two operands, and each of
these may be either an AP coming from the input trace, or from a queue (which is
a result produced by another PE). 
It provides its results by
modifying the contents of a particular Que with ID $r\_{qid}$. 
Finally, since the number of PEs
    is fixed at design time, it is possible that some PEs may not be used for
    evaluating the current input MTL formula. Such PEs are deactivated by
    setting the {\em isActive} bit to $0$. The rest of the instruction bytes are
    ignored if this bit is set to $0$. The {\em isActive} bit is also sent
	to the PE2Q interconnect so that it can ignore outputs from a PE (stale 
	values on the pins) when it is inactive.

\begin{figure}[htbp]
		\centering
		\includegraphics[width=\textwidth]{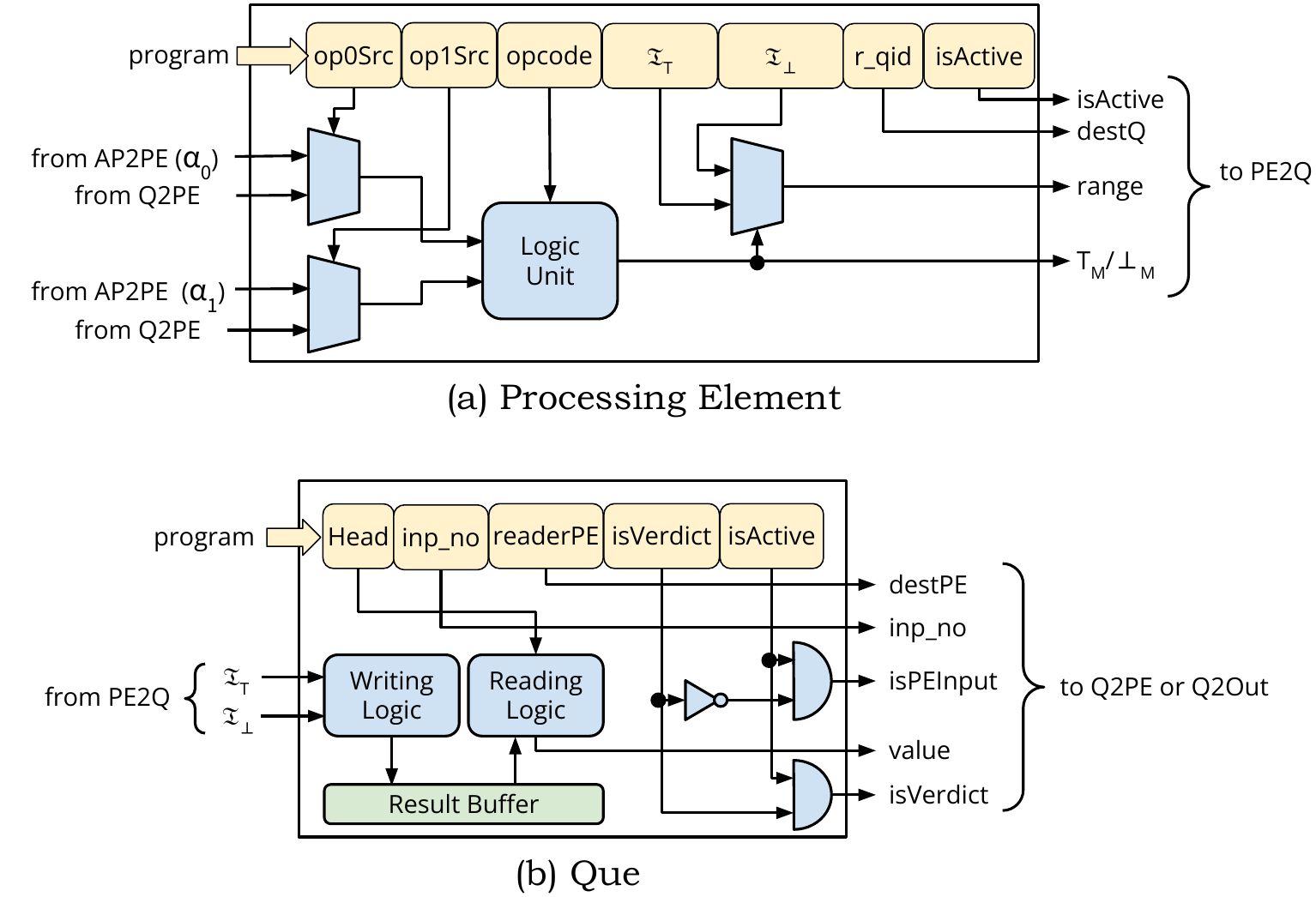}
		\caption{Architecture of a (a) Processing Element, and (b) Que} \label{fig:pe_q}
\end{figure}

Figure~\ref{fig:pe_q} shows the architecture of a PE. The datapath of a PE 
starts from the input pins coming in either
from AP2PE or Q2PE interconnects, for both the operands. The inputs go into the logic
unit which is capable of performing {\em all} the five basic operations (see
opcodes in Table~\ref{tab:pe_program}). The result of the operation is used as
the select line in a multiplexer. If the result of the operation is $\true$, the
programmed $\mathcal{I}_\true$ is sent to the output pins that are read
by the PE2Q interconnect. If the result is $\false$, the $\mathcal{I}_\false$
is sent to the output pins. The $r\_qid$ field is used as the
destination field by the interconnect so as to forward the values to the
appropriate queue.

\subsection{Que}
\label{sec:q}
Figure~\ref{fig:pe_q} shows the architecture of a Q. The Que's operation
follows the description given in Section~\ref{sec:am_operators}. In each cycle,
the que receives the $\mathcal{I}_\true$ and $\mathcal{I}_\false$ intervals
from the PE2Q interconnect, and modifies the contents of its buffer
accordingly. A Q can be programmed as indicated in Table~\ref{tab:q_program}.
The values deleted from a queue form the source operand of one or more PEs, or
it forms the verdict of the formula ($isVerdict$ is programmed to $\true$). In
the former case, the value is sent to the Q2PE interconnect, and in the latter
case, it is sent to the Q2Out interconnect to be communicated to the external
world. The position at which the deletion must be done is programmable.

\begin{table}[htbp]
	\scriptsize
	\centering
	\caption{Programming of a Que} \label{tab:q_program}
\setlength{\tabcolsep}{4pt} 
\renewcommand{\arraystretch}{1.35}
	\begin{tabular}{|c|p{8cm}|c|}
		\hline
		\textbf{Field} & \multicolumn{1}{c|}{\textbf{Description}} & \textbf{\# bits} \\
		\hline
		{\em isActive} & Is the Q being used & 1 \\
		\hline
		{\em isVerdict} & Are the queue contents the overall monitor verdicts & 1 \\
		\hline
		{\em readerPE} & ID of the PE that is reading the queue contents & $ \lceil \log_2N_{PE} \rceil$ \\
		\hline
		{\em inp\_no} & Are the queue contents the first or second operand of {\em readerPE} & 1 \\
		\hline
		{\em Head} & deletion position & $\lceil \log_2 Q_{SZ} \rceil$ \\
		\hline
	\end{tabular}
\end{table}

\subsection{Programmable Interconnects}
\label{sec:interconnects}
All four interconnects are implemented as single-hop crossbar
switches. This implies that in each interconnect, all sender modules 
have dedicated paths to all receiver modules, and for each sender, 
one of the paths may be chosen at runtime by programming
the switch. Since several PEs can forward
their results ($\mathcal{I}_\true$/$\mathcal{I}_\false$ and 
$\true_\maybe$/$\false_\maybe$) to the same Q, as in the case of the $\until$ operator 
(see Table~\ref{tab:am_operators}), the PE2Q interconnect coalesces the 
inputs from the several PEs before forwarding them to the destination Q. 
Such coalescing involves merging the $\mathcal{I}_\true$ 
(equivalently, $\mathcal{I}_\false$) intervals from different PEs into a single, contiguous
$\mathcal{I}_\true$ ($\mathcal{I}_\false$) interval. This
coalesced interval is the input of the destination Q. It is guaranteed that
coalescing will only increase the range, and that the values to be written
in these ranges will not conflict with each other 
(see Proposition~\ref{prop:simul_modify}), and hence the coalescing operation
 is straightforward.

\subsection{Monitor Programming}
\label{sec:example_monitor}

Our compiler follows a simple strategy. It beings by constructing the AST of the
input MTL formula. Then, we process the nodes in reverse topological order,
assigning one or more PEs to each node, and one queue to the single outgoing
edge of each node. We perform queue balancing as described in
Section~\ref{sec:em_composing}. Finally, the appropriate bits to program the
PEs, Qs and the interconnects are generated as discussed earlier.

	\begin{table}
		\scriptsize
		\centering
		\caption{Programming of different monitor components} \label{tab:example_monitor}
  \setlength{\tabcolsep}{4pt} 
  \renewcommand{\arraystretch}{1.35}
		\begin{tabular}{c}
			\begin{minipage}{\textwidth}
				\centering
				\begin{tabular}{|c|c|c|c|c|c|c|c|l|}
					\hline
					\multicolumn{9}{|c|}{\textbf{Programming of PEs}} \\
					\hline
					PE & {\em isActive} & {\em op0Src} & {\em op1Src} & {\em opcode} & {\em r\_qid} & $\mathcal{I}_\true$ & $\mathcal{I}_\false$ & \multicolumn{1}{c|}{Remarks} \\
					\hline
					$PE_0$ & true & AP & X & {\em not} & 0 & [0, 0] & [0, 0] & $EM_4$ from Fig~\ref{fig:am_example}(b) \\
					\hline
					$PE_1$ & true & AP & X & {\em wire} & 1 & [1, 4] & [4, 4] & $EM_3$ from Fig~\ref{fig:am_example}(b) \\
					\hline
					$PE_2$ & true & Q2PE & X & {\em wire} & 2 & [0, 1] & [1, 1] & $EM_2$ from Fig~\ref{fig:am_example}(b) \\
					\hline
					$PE_3$ & true & Q2PE & Q2PE & {\em or} & 3 & [0, 0] & [0, 0] & $EM_1$ from Fig~\ref{fig:am_example}(b) \\
					\hline
					$PE_{4-7}$ & false & X & X & X & X & X & X & unused \\
					\hline
				\end{tabular}
			\end{minipage}
			\\
			\begin{minipage}{\textwidth}
				\centering
  \setlength{\tabcolsep}{4pt} 
  \renewcommand{\arraystretch}{1.35}
				\begin{tabular}{|c|c|c|c|c|c|l|}
					\hline
					\multicolumn{7}{|c|}{\textbf{Programming of Qs}} \\
					\hline
					Q & {\em isActive} & {\em isVerdict} & {\em readerPE} & {\em inp\_no} & {\em Head} & \multicolumn{1}{c|}{Remarks} \\
					\hline
					$Q_0$ & true & false & 2 & 0 & 1 & $Q_4$ from Fig~\ref{fig:am_example}(b) \\
					\hline
					$Q_1$ & true & false & 3 & 1 & 5 & $Q_3$ from Fig~\ref{fig:am_example}(b) \\
					\hline
					$Q_2$ & true & false & 3 & 0 & 3 & $Q_2$ from Fig~\ref{fig:am_example}(b) \\
					\hline
					$Q_3$ & true & true & X & X & 1 & $Q_1$ from Fig~\ref{fig:am_example}(b) \\
					\hline
					$Q_{4-7}$ & false & X & X & X & X & unused \\
					\hline
				\end{tabular}
			\end{minipage}
		\end{tabular}
	\end{table}

We elucidate the working of the compiler by using the example considered in
Section~\ref{sec:em_composing}.  Table~\ref{tab:example_monitor} gives the
programming of the different monitor components. It assumes a monitor with
$N_{PE} = 8$, $N_Q = 8$, and $N_{AP} = 4$ (supports 4 APs $\sigma_0$ to
$\sigma_3$).

The number of bits required to program the PEs is $N_{PE} \times (6 + \lceil \log_2N_Q \rceil
+ 4 \times \lceil \log_2 Q_{SZ} \rceil )$. The number of bits required to program the Qs
is $N_Q \times (3 + \lceil \log_2N_{PE} \rceil + \lceil \log_2 Q_{SZ} \rceil)$. The number of bits
required to program the AP2PE interconnect is $N_{PE}\times 2 \times
\lceil \log_2N_{AP} \rceil$.

\section{Tool Description and Evaluation}
\label{sec:tool}

Since the design of the MTL monitor is parameterized by $N_{PE}$, $N_Q$,
$N_{AP}$ and $Q_{SZ}$, we have developed a tool that takes these parameters as
input and generates the corresponding design of the programmable monitor in the
Clash Hardware Description Language~\cite{clash}.  The Clash compiler is then
used to generate a synthesizable Verilog code from the generated Clash
description of the programmable monitor.  Our tool also generates the program
bits for the PEs and the Qs when the monitor needs to be programmed to monitor
a given formula. The program bits are sent into the MTL monitor 8-bits at a
time in consecutive clock cycles. This is done to restrict the number of I/O
pins in the design, but does not affect the functionality of circuit.
The tool also gives the verdict latency -- the difference
between the time a verdict is available and the time the verdict corresponds
to. The tool is available online and 
open-sourced~\cite{hebballi_2026_19084848}.

Our tool is coded in Python. The correctness of the generated HDL was verified
by simulating the design against several MTL properties. The input properties
were changed dynamically as exemplified in the Figure~\ref{fig:waves}
(mimicking runtime reprogramming). The testbench was also written using Clash.

We used the Synopsys Design Compiler (DC) and the $32nm$ Standard Cell Library
available in the Educational Design Kit (EDK) (SAED32nm\_EDK) for synthesizing
the design.

\subsection{Simulations}

\begin{figure}[htbp]
  \centering
  \includegraphics[width=0.85\textwidth]{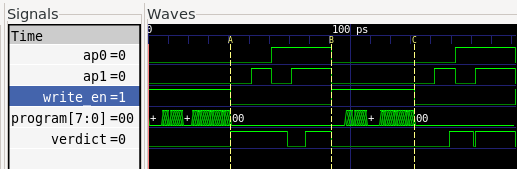}
  \caption{Simulation showing the MTL monitor being reprogrammed at runtime} \label{fig:waves}
\end{figure}

Figure~\ref{fig:waves} shows the waveform corresponding to a scenario the 
monitor is reprogrammed dynamically (at runtime) to monitor two different 
properties one after another. The Value Change Dump (VCD) file is
generated from a Clash simulation of a design $N_{PE} = N_Q = 4$, $N_{AP} = 8$
and $Q_{SZ} = 16$. The VCD file is visualized using GTKWave. The Clock, Reset
and Enable signals of a design are implicit in Clash, and hence are not seen in
the waveform. We have used the default Clash setting of a $1 ps$ clock period for
the simulations. The choice of clock period is irrelevant during
simulations because we are only interested in verifying the behavior in each clock
cycle. The design itself has been pipelined to increase the post-synthesis clock
frequency by adding flip-flops at the inputs of the interconnect.

The first property monitored in this scenario is $AP_0 \implies \next AP_1$,
and is followed by $AP_0 \lor \diamond_{[1,3]} AP_1$. The number of cycles to
reprogram the monitor depends on the choice of parameters such as $N_{PE},
N_{Q}$ and $N_{AP}$. For our simulated design, it takes 41 cycles to reprogram
the monitor because only 8 program bits are sent into the monitor in each clock
cycle (see the signal {\tt program[7:0]}). We observe that the cycles from
start of simulation (cycle 0) to Marker A (vertical dashed line at cycle 40)
programs the monitor to verify the first property, and cycles 91 to 131
reprograms the monitor to verify the second property (between Markers B and C).
The signal {\tt write\_en} is high during these periods.
The MTL monitor verifies the first property from cycle $41$ to cycle $90$, which lies
in between the Markers A and B. As per the semantics of the first property,
the verdict is \texttt{false} when $AP_0$ is \texttt{true} and $AP_1$ is
\texttt{false} in the next cycle, and the verdict is \texttt{true} otherwise.
In Figure~\ref{fig:waves}, we see that $AP_0$ is \texttt{false} from cycle $41$
to $61$, and in cycle $62$, it goes to \texttt{true}.  $AP_1$ is \texttt{false}
in cycle $63$. Therefore, the verdict for the formula is \texttt{false} for
cycle $62$.  This verdict is available at cycle $70$. Likewise, the verdict for
cycle $63$ (also \texttt{false}) is available at cycle $71$, and so on.  Thus,
the latency in this case is $8$ cycles (verdict for cycle $i$ is available in
cycle $(i+8)$), while the throughput is $1$ verdict per cycle.  This latency is
a function of the horizon of the formula being monitored. We can make similar
observations in the duration between Marker C (at $132ps$) and the end of
simulation (at $182ps$) when the second property is being monitored. In this
case, the verdicts are observed after a latency of $11$ cycles.

This shows that the generated MTL monitor can be reprogrammed only by changing
the values at its input pins without requiring its internal circuits to be
re-synthesized (akin to a general purpose processor).

\subsection{Synthesis Results}
\label{ssec:results}

\begin{figure}[htbp]
\centering
\begin{tabular}{c}
  \begin{minipage}{0.95\textwidth}
    \includegraphics[width=\textwidth]{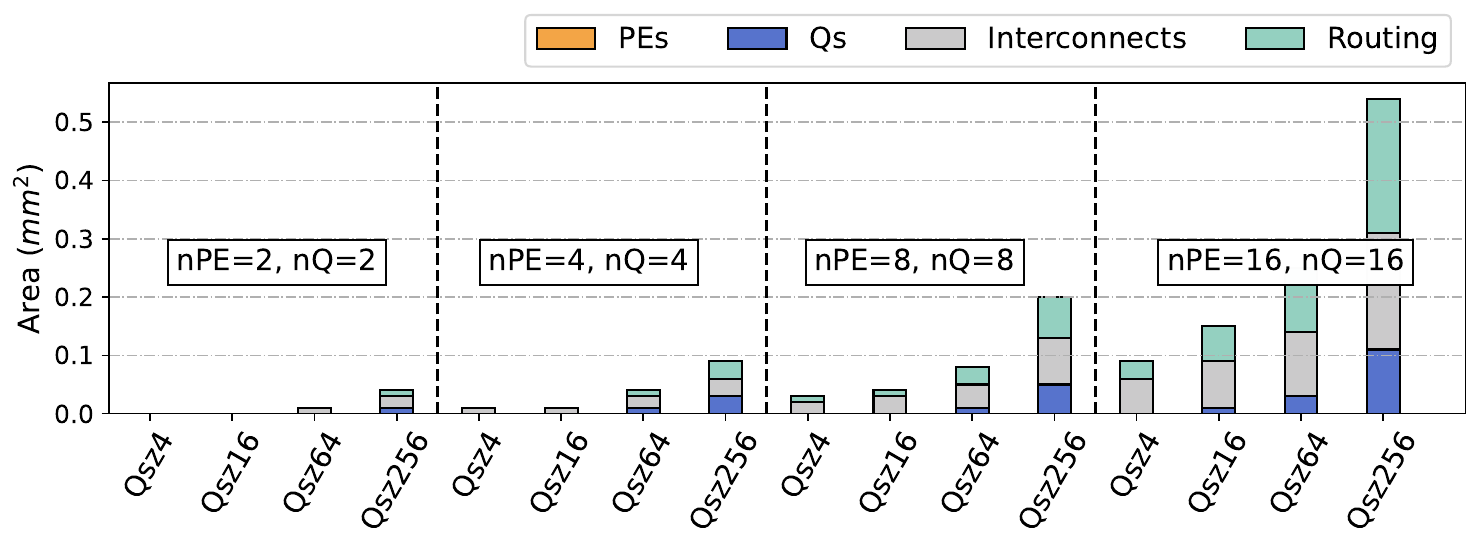}
    \caption{Area occupied by PEs, Qs and interconnects (for $N_{AP} = 16$)} \label{fig:nAP16_sqmm}
  \end{minipage}
  \\
  \begin{minipage}{0.95\textwidth}
    \includegraphics[width=\textwidth]{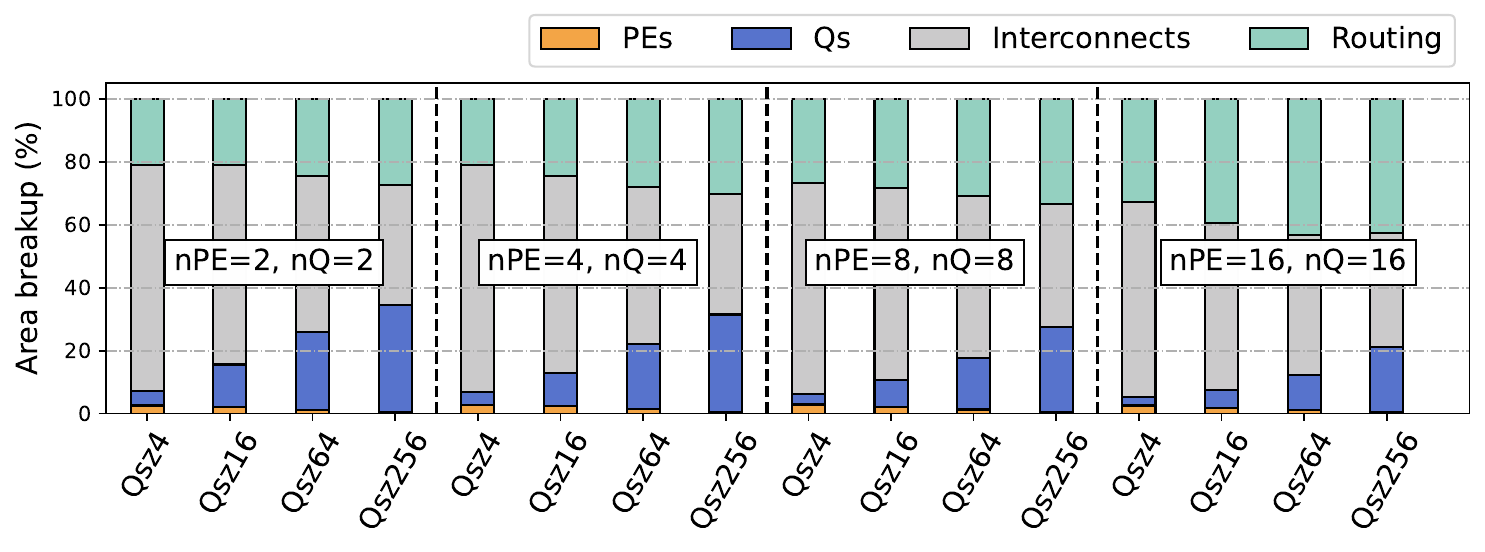}
    \caption{Percentage-wise area breakup of MTL Monitors (for $N_{AP} = 16$)} \label{fig:nAP16_percent}
  \end{minipage}
\end{tabular}
\end{figure}

We performed a design space exploration of a hardware monitor (supporting 
$N_{AP} = 16$) by varying $N_{PE} (2,4,8,16)$, $N_Q (2,4,8,16)$ and 
$Q_{SZ} (4,16,64,256)$ . $N_Q$ is set to $N_{PE}$ in all our experiments. Figure~\ref{fig:nAP16_sqmm} shows the total area occupied by 
each of the synthesized MTL monitor, along with its respective component-wise breakup.
Figure~\ref{fig:nAP16_percent} shows the percentage-wise area breakup of the same. 
The bars labeled {\em PEs}, {\em Qs}, and {\em Interconnects} capture 
the area occupied by the standard cells by these components, and the bar 
corresponding to {\em Routing} captures the area occupied by the wires connecting
these standard cells. We have grouped the standard cells correponding to 
components such as AP2PE, PE2Q, Q2PE and Q2OUT under Interconnects for brevity.

The total area scales superlinearly with increase in $N_{PE}$ (keeping $Q_{SZ}$
constant).  This is because the dominant contributors towards the total area --
Routing and Interconnects -- increase superlinearly. The reasons for this are
that with increase in $N_{PE}$, (i) the Interconnect has to support more input
and output ports, and (ii) the overall design is larger resulting in longer
wires. This sharp increase in area can be mitigated to some extent by using
better interconnect architectures instead of single-hop crossbar switches. The
exploration of better interconnect design is left for future work. In general,
such exploration of the design space is a standard practice in hardware design
and helps designers fine-tune the circuit to meet requirements of typical input
properties. Such exploration is a one-time (design) effort and need not be
repeated for every input property.  With increase in $Q_{SZ}$, we observe that
while the area occupied by the Ques increases superlinearly, the total area
increases sublinearly as this is the trend seen in the dominant Routing and
Interconnect components.

We also note that the largest design corresponding to the case 
$N_{PE}=N_Q=16$, $N_{AP}=16$ and $Q_{SZ}=256$ occupies $0.55 mm^2$ 
($547781 \mu m^2$) which is just $0.13$\% of a typical $400 mm^2$ die 
(usually that of a general purpose processor). 

Our synthesized monitors has a critical path delay of $0.8 ns$ (resulting in a
clock frequency of $1.25 GHz$). This effectively means that the monitor has a
high throughput -- it observes $1.25 \times 10^9$ observations per second, with
one verdict produced per cycle. This is significantly better than the typical
clock frequencies achievable (in $MHz$) using FPGAs.  The clock frequency can
be increased further (if required) by adding more flip-flops along the critical
path. This will increase the pipeline depth, which increases the number of
clock cycles required to get the verdict for a given input event (that is, the
latency), and the overall area occupied by the circuit. However, the throughput
continues to remain one verdict per cycle (when increasing the pipeline depth).
Since the design (or circuit) does not change with changing input properties,
such trade-offs between area and frequency needs to be explored only once (at
design time).

\section{Conclusion}
\label{sec:conc}

In this work, we proposed a hardware MTL monitor which is amenable to be
implemented using standard cells, while still being flexible to handle changes
in the properties to be monitored. This is achieved by using a set of
programmable PEs, Qs and interconnects. Each PE implements five basic operations
and a set of Que update rules, while the result buffer present inside a Q is
responsible for retaining timing behavior and holding the intermediate results
and verdicts generated. The interconnects are programmed to carry the results
between PEs and Qs and hence allows them to be composed to monitor complex MTL
properties. We develop a monitor Generator that takes as input parameters such
as $N_{PE}$ and $N_Q$ and generates the synthesizable verilog code of the
proposed monitor. We also developed an associated Compiler that takes the
property to monitor in addition to the aforementioned parameters and generates
program bits to program the monitor appropriately. Our simulations show that
generated designs can be programmed at runtime. We also observe that even a
fairly large MTL monitor which supports up to $16$ APs occupies only a modest
area of $0.55 mm^2$ while operating at $1.25 GHz$.

\bibliography{refs}
\appendix
\newpage  
\section*{Appendix : Proofs of correctness of Evaluator Machines}
\label{subsec:am_mtl_proofs} 
 
\begingroup
\scriptsize
\def\arraystretch{1.5}

\begin{longtable}{|c|c|c|c|p{5cm}|}
	\caption{Evaluator Machine operation corresponding to different MTL operators} \label{tab:am_operators} \\
		\toprule
		\multicolumn{2}{|c|}{\textbf{Operand value}} & \multicolumn{2}{c|}{\textbf{Queue modification}} & \multirow{3}{*}{\hspace{2.5cm}\textbf{Reasoning}} \\
		\multicolumn{2}{|c|}{\textbf{at time $i$}} & \multicolumn{2}{c|}{} & \\
		\cmidrule(lr){0-3}
		$\hspace{0.45cm}\alpha_0\hspace{0.45cm}$ & \hspace{0.45cm}$\alpha_1$\hspace{0.45cm} & \textbf{Value} & \textbf{Position(s)} & \\
    \hline
    \hline
		\multicolumn{5}{|c|}{\textbf{Not operator : $\neg \alpha_0$}} \\
    \hline
		$\false$ & - & $\true$ & $0$ & \multirow{2}{*}{$r^i = \neg \alpha_0^i$} \\
		\cmidrule(lr){0-3}
		$\true$ & - & $\false$ & $0$ & \\
    \hline
    \hline
		\multicolumn{5}{|c|}{\textbf{Or operator : $\alpha_0 \lor \alpha_1$}} \\
    \hline
		$\false$ & $\false$ & $\false$ & $0$ & \multirow{4}{*}{$r^i = \alpha_0^i \lor \alpha_1^i$} \\
		\cmidrule(lr){0-3}
		$\false$ & $\true$ & $\true$ & $0$ & \\
		\cmidrule(lr){0-3}
		$\true$ & $\false$ & $\true$ & $0$ & \\
		\cmidrule(lr){0-3}
		$\true$ & $\true$ & $\true$ & $0$ & \\
    \hline
    \hline
		\multicolumn{5}{|c|}{\textbf{And operator : $\alpha_0 \land \alpha_1$}} \\
    \hline
		$\false$ & $\false$ & $\false$ & $0$ & \multirow{4}{*}{$r^i = \alpha_0^i \land \alpha_1^i$} \\
		\cmidrule(lr){0-3}
		$\false$ & $\true$ & $\false$ & $0$ & \\
		\cmidrule(lr){0-3}
		$\true$ & $\false$ & $\false$ & $0$ & \\
		\cmidrule(lr){0-3}
		$\true$ & $\true$ & $\true$ & $0$ & \\
    \hline
    \hline
		\multicolumn{5}{|c|}{\textbf{Implies operator : $\alpha_0 \implies \alpha_1$}} \\
    \hline
		$\false$ & $\false$ & $\true$ & $0$ & \multirow{4}{*}{$r^i = \alpha_0^i \implies \alpha_1^i$} \\
		\cmidrule(lr){0-3}
		$\false$ & $\true$ & $\true$ & $0$ & \\
		\cmidrule(lr){0-3}
		$\true$ & $\false$ & $\false$ & $0$ & \\
		\cmidrule(lr){0-3}
		$\true$ & $\true$ & $\true$ & $0$ & \\
    \hline
    \hline
		\multicolumn{5}{|c|}{\textbf{Next operator : $\next\; \alpha_0$}} \\
    \hline
		$\false$ & - & $\false$ & $1$ & \multirow{2}{*}{$r^{i-1} = \alpha_0^i$} \\
		\cmidrule(lr){0-3}
		$\true$ & - & $\true$ & $1$ & \\
    \hline
    \hline
		\multicolumn{5}{|c|}{\textbf{Box operator : $\Box_{[t_1, t_2]}\; \alpha_0$}} \\
    \hline
		$\false$ & - & $\false$ & $[t_1, t_2]$ & $\alpha_0^i = \false \implies r^t = \false, \forall t \in [i-t_2, i-t_1]$ \\
    \hline
		$\true$ & - & $\true_\maybe$ & $t_2$ & $r^{i-t_2} \neq \false \land \alpha_0^i=\true \implies r^{i-t_2}=\true$ \\
    \hline
    \hline
		\multicolumn{5}{|c|}{\textbf{Diamond operator : $\diamond_{[t_1, t_2]}\; \alpha_0$}} \\
    \hline
		$\false$ & - & $\false_\maybe$ & $t_2$ & $r^{i-t_2} \neq \true \land \alpha_0^i=\false \implies r^{i-t_2}=\false$ \\
    \hline
		$\true$ & - & $\true$ & $[t_1, t_2]$ & $\alpha_0^i = \true \implies r^t = \true, \forall t \in [i-t_2, i-t_1]$ \\
    \hline
    \hline
		\multicolumn{5}{|c|}{\textbf{Until operator : $\alpha_0 \;\until_{[t_1, t_2]}\; \alpha_1$}} \\
    \hline
		$\false$ & $\false$ & $\false_\maybe$ & $t_2$ & [$M_2: wire\; \alpha_1$] $\alpha_1$ has never been true in the $[t_1, t_2]$ window for the time $i-t_2$ \\
		 & & $\false_\maybe$ & $[t_1, t_2-1]$ & [$M_3: or\; \alpha_0,\alpha_1$] $\forall t \in [i-t_2+1, i-t_1], r^t \neq \true \land \alpha_0^i = \false \land \alpha_1^i = \false \implies r^t = \false$ \\
		 & & $\false$ & $[0, t_1-1]$ & [$M_1: wire\; \alpha_0$] $\alpha_0 = \false$ has occurred in the $[0, t_1-1]$ window for times $[i, i-t_1+1]$ \\
    \hline
		$\false$ & $\true$ & $\false$ & $[0, t_1-1]$ & [$M_1: wire\; \alpha_0$] $\alpha_0 = \false$ has occurred in the $[0, t_1-1]$ window for times $[i, i-t_1+1]$ \\
		 & & $\true_\maybe$ & $[t_1, t_2]$ & [$M_2: wire\; \alpha_1$] $\forall t \in [i-t_2, i-t_1], r^t \neq \false \land \alpha_1^i=\true \implies r^t=\true$ \\
    \hline
		$\true$ & $\false$ & $\false_\maybe$ & $t_2$ & [$M_2: wire\; \alpha_1$] $\alpha_1$ has never been true in the $[t_1, t_2]$ window for the time $i-t_2$ \\
    \hline
		$\true$ & $\true$ & $\true_\maybe$ & $[t_1, t_2]$ & [$M_2: wire\; \alpha_1$] $\forall t \in [i-t_2, i-t_1], r^t \neq \false \land \alpha_1^i=\true \implies r^t=\true$ \\
    \hline
    \hline
		\multicolumn{5}{|c|}{\textbf{Until operator : $\alpha_0 \;\until_{[0, t_2]}\; \alpha_1$}} \\
    \hline
		$\false$ & $\false$ & $\false_\maybe$ & $t_2$ & [$M_2: wire\; \alpha_1$] $\alpha_1$ has never been true in the $[0, t_2]$ window for the time $i-t_2$ \\
		 & & $\false_\maybe$ & $[0, t_2-1]$ & [$M_1: or\; \alpha_0,\alpha_1$] $\forall t \in [i-t_2+1, i], r^t \neq \true \land \alpha_0^i = \false \land \alpha_1^i = \false \implies r^t = \false$ \\
    \hline
		$\bot$ & $\top$ & $\true_\maybe$ & $[0, t_2]$ & [$M_2: wire\; \alpha_1$] $\forall t \in [i-t_2, i], r^t \neq \false \land \alpha_1^i=\true \implies r^t=\true$ \\
    \hline
		$\top$ & $\bot$ & $\false_\maybe$ & $t_2$ & [$M_2: wire\; \alpha_1$] $\alpha_1$ has never been true in the $[0, t_2]$ window for the time $i-t_2$ \\
    \hline
		$\top$ & $\top$ & $\true_\maybe$ & $[0, t_2]$ & [$M_2: wire\; \alpha_1$] $\forall t \in [i-t_2, i], r^t \neq \false \land \alpha_1^i=\true \implies r^t=\true$ \\
    \hline
    \hline
		\multicolumn{5}{|l|}{$\true_\maybe$: write value $\true$, if current value of cell is $\maybe$} \\
		\multicolumn{5}{|l|}{$\false_\maybe$: write value $\false$, if current value of cell is $\maybe$} \\
		\multicolumn{5}{|l|}{$\alpha_k^t$: value of operand $\alpha_k$ at time $t$} \\
		\multicolumn{5}{|l|}{$r^t$: result corresponding to time $t$, as computed up to the current point in time} \\
    \hline
\end{longtable}

Table~\ref{tab:am_operators} describes the operating logic followed by the EM
corresponding to each of the MTL operators. For all possible combinations of
input operand values that the AM may receive in a certain cycle $i$
($\alpha_{\{0,1\}}^i$), column 4 shows the indices of the result queue within
the EM that are changed to the value shown in column 3, to record the effect
$\alpha_{\{0,1\}}^i$ has on a (partial) verdict. The reasoning given in column 5
helps intuitively prove that the EM's behavior accurately captures the
corresponding operator.

We now give formal proofs for each operator individually:

\subsubsection*{(Proof of Theorem~\ref{lab-correctness-thm} for Boolean operators)}:        
By definition, $l()$ maps all the Boolean operators to $1$. We
prove, for all time instances $i \geq 1$:
\begin{itemize}
	\item {\em if} $x^{i-1} \entails \phi$ {\em, then} $Q^{i}[1]=\true$.  
	\item {\em if} $x^{i-1} \doesnotentail \phi$ {\em, then} $Q^{i}[1]=\false$.  
\end{itemize}

\begin{proof}[$\phi =\neg \alpha$]:
	Assume $x^{i-1} \entails \neg \alpha$. This is possible only when,
at time step $i-1$, formula $\alpha$ evaluates to $\false$. 
	In this case AM sets $Q^{i-1}[0]=\true$, which is the verdict $r^{i-1}$
for $\neg \alpha$. 

	Assume $x^{i-1} \doesnotentail \neg \alpha$. This is possible only when,
at time step $i-1$, formula $\alpha$ evalutes to $\true$. 
In this case AM sets $Q^{i-1}[0]=\false$, which is the verdict $r^{i-1}$
for $\neg \alpha$.

In the next time step, when
	$\maybe$ is inserted, the verdict $r^{i-1}$ moves to $Q^{i}[1]$ as required,
	as the AM does not modify cell $1$.

\end{proof}  
 
\begin{proof}[$\phi = \alpha \land \beta$]:
	Assume $x^{i-1} \entails \alpha \land \beta$. This is possible only when,
at time instance $i-1$, both formulas $\alpha$  and $\beta$ evaluate to $\true$. 
	In this case AM sets $Q^{i-1}[0]=\true$, which is the verdict $r^{i-1}$
for $\alpha \land \beta$.
 
	Assume $x^{i-1} \doesnotentail \alpha \land \beta$. This is possible only when,
at time instance $i-1$, at least one of the formulas of $\alpha$ and
$\beta$  evaluates to $\false$. 
	In this case AM sets $Q^{i-1}[0]=\false$, which is the verdict $r^{i-1}$
for $\alpha \land \beta$. 

In the next time step, when
	$\maybe$ is inserted, the verdict $r^i$ moves to $Q^{i}[1]$ as required,
	as the AM does not modify cell $1$.
\end{proof}

\subsubsection*{(Proof of Theorem~\ref{lab-correctness-thm} for $\next \alpha$)}:        
\begin{proof}
	By definition, $l()$ maps the $\next$ operator to $2$.\\
Here we prove that, for all time instances $i$, with $i \geq 2$,  
\begin{center} 
	{\em if} $x^{i-2} \entails \next \alpha${\em, then} $Q^{i}[2]=\true$.  
\end{center} 
\begin{center} 
	{\em if} $x^{i-2} \doesnotentail \next \alpha${\em, then} $Q^{i}[2]=\false$.  
\end{center} 

By the definition of the AM, at time instance $i-2$, 
queue $Q^{i-2}[0]$ contains $\maybe$. This is 
the verdict corresponding to time $i-2$, 
and has not been resolved by the AM at this point in time.
 
Assume $x^{i-2} \entails \next \alpha$. This means $x^{i-1} \entails \alpha$.
	So AM sets $Q^{i-1}[1]=\true$, which is the verdict $r^{i-2}$ for $\next \alpha$
for instance $i-2$. 
 
Assume $x^{i-2} \doesnotentail \next \alpha$. This means $x^{i-1}
\doesnotentail \alpha$.
	So AM sets $Q^{i-1}[1]=\false$, which is the verdict $r^{i-2}$ for $\next \alpha$
for instance $i-2$. 

In the next time step, when
	$\maybe$ is inserted, the verdict $r^{i-2}$ moves to $Q^{i}[2]$ as required,
	as the AM does not modify cell $2$.

\end{proof}

\subsubsection*{(Proof of Theorem~\ref{lab-correctness-thm} for $\Box_{[t_1,t_2]} \alpha$)}:        
\begin{proof}
	By definition, $l()$ maps the $\Box_{[t_1,t_2]}$ operator to $t_2+1$.\\
To prove, for all time instances $i \geq t_2+1$:
\begin{itemize}
	\item {\em if} $x^{i-t_2-1} \entails \Box_{[t_1,t_2]}\alpha${\em, then} $Q^{i}[t_2+1]=\true$.  
	\item {\em if} $x^{i-t_2-1} \doesnotentail \Box_{[t_1,t_2]}\alpha${\em, then} $Q^{i}[t_2+1]=\false$.  
\end{itemize}
 
(\emph{Assume $x^{i-t_2-1} \entails \Box_{[t_1,t_2]}\alpha$}):
Then we have for all $l$ in $[i-1-t_2+t_1,i-1],~x^{l} \entails \alpha$.

We break the interval $[i-1-t_2,i-1]$ into four parts:
\begin{itemize}
	\item at time $i-t_2-1$, as per the definition of the AM, a value of $\maybe$
		is inserted into the queue at index $0$, and this value of $Q^{i-1-t_2}[0]$
		is not altered in this time unit. This value corresponds to the verdict $r^{i-1-t_2}$.
	\item for all times $l' \in [i-1-t_2+1, i-1-t_2+t_1-1]$, as per the definition of
		the AM, the values in cells $[0, t_1-1]$ are not altered. The
		value $\maybe$ which was set in $Q^{i-1-t_2}[0]$ simply shifts through
		$Q^{l'}[l'-(i-1-t_2)]$ unaltered.
	\item for all times $l'' \in [i-1-t_2+t_1, i-2]$, since $x^{l''} \entails \alpha$, as per the definition of
		the AM, the values in cells $[0, t_2-1]$ are not altered. The
		value $\maybe$ which was set in $Q^{i-1-t_2}[0]$ simply shifts through
		$Q^{l''}[l''-(i-1-t_2)]$ unaltered.
	\item at time $i-1$, since $x^{i-1} \entails \alpha$, as per the definition
		of the AM, the value in cell $Q^{i-1}[t_2]$ is set to $\true$.
\end{itemize}

(\emph{Assume $x^{i-t_2-1} \doesnotentail \Box_{[t_1,t_2]}\alpha$}):
Then we have for some $l$ in $[i-1-t_2+t_1,i-1],~x^{l} \doesnotentail \alpha$.

We break the interval $[i-1-t_2,i-1]$ into five parts:
\begin{itemize}
	\item at time $i-t_2-1$, as per the definition of the AM, a value of $\maybe$
		is inserted into the queue at index $0$, and this value of $Q^{i-1-t_2}[0]$
		is not altered in this time unit. This value corresponds to the verdict $r^{i-1-t_2}$.
	\item for all times $l' \in [i-1-t_2+1, i-1-t_2+t_1-1]$, as per the definition of
		the AM, the values in cells $[0, t_1-1]$ are not altered. The
		value $\maybe$ set in $Q^{i-1-t_2}[0]$ simply shifts through
		$Q^{l'}[l'-(i-1-t_2)]$ unaltered.
	\item for all times $l'' \in [i-1-t_2+t_1, l-1]$, since $x^{l''} \entails \alpha$, as per the definition of
		the AM, the values in cells $[0, t_2-1]$ are not altered. The
		value $\maybe$ which was set in $Q^{i-1-t_2}[0]$ simply shifts through
		$Q^{l''}[l''-(i-1-t_2)]$ unaltered.
	\item at time $l$, since $x^{l} \doesnotentail \alpha$, as per the definition
		of the AM, the value in cells $Q^{l}[t_1, t_2]$ is set to $\false$.
		This range includes the cell $Q^{l}[l-(i-1-t_2)]$, whose old value was $\maybe$.
	\item for all times $l''' \in [l+1, i-1]$, as per the definition of the AM, the value
		$\false$ which was set in $Q^{l}[l-(i-1-t_2)]$ simply shifts through $Q^{l'''}[l'''-(i-1-t_2)]$ unaltered.
		Thus, at time $i-1$, the value in cell $Q^{i-1}[t_2]$ is $\false$.
\end{itemize}

In the next time step, when
	$\maybe$ is inserted, the verdict $r^{i-1-t_2}$ moves to $Q^{i}[t_2+1]$ as required,
	as the AM does not modify cell $t_2+1$.
\end{proof} 

	Theorem~\ref{lab-correctness-thm} for $\phi=\Diamond_{[t_1,t_2]} \alpha$
	can be proven in a similar manner.

\subsubsection*{(Proof of Theorem~\ref{lab-correctness-thm} for $\phi=
\alpha \until_{[t_1,t_2]} \beta$)}:

\begin{proof} 
	By definition, $l()$ maps the $\until_{[t_1,t_2]}$ operator to $t_2+1$.\\
To prove, for all time instances $i \geq t_2+1$:
\begin{itemize}
	\item {\em if} $x^{i-t_2-1} \entails \alpha \until_{[t_1,t_2]} \beta${\em, then} $Q^{i}[t_2+1]=\true$.  
	\item {\em if} $x^{i-t_2-1} \doesnotentail \alpha \until_{[t_1,t_2]} \beta${\em, then} $Q^{i}[t_2+1]=\false$.  
\end{itemize}
 
(\emph{Assume $x^{i-t_2-1} \entails \alpha \until_{[t_1,t_2]} \beta$}):
 
Let $l$ be the earliest time instance in $[i-1-t_2+t_1,i-1]$ such that $x^{l} \entails \beta$ and 
for all $k$ such that  $i-1-t_2 \leq k< l-1, ~x^{l} \entails \alpha$.
We break the interval $[i-1-t_2,i-1]$ into five parts:
\begin{itemize}
\item At time $i-1-t_2$, as per the definition of the AM, a value of $\maybe$
		is inserted into the queue at index $0$, and this value of $Q^{i-1-t_2}[0]$
		is not altered in this time unit. This value corresponds to the verdict $r^{i-1-t_2}$.
\item For all times $l' \in [i-1-t_2+1, i-1-t_2+t_1-1]$, we know that 
      $x^{l'} \entails \alpha$ and as per the definition of the AM, 
      the values in cells $[0, t_1-1]$ are not altered. 
      The value $\maybe$ which was set in $Q^{i-1-t_2}[0]$ 
      simply shifts through $Q^{l'}[l'-(i-1-t_2)]$ unaltered.
\item For all times $l'' \in [i-1-t_2+t_1, l-1]$, we know that 
	$x^{l''} \entails \alpha$ and $x^{l''} \doesnotentail \beta$.
	As per the definition of the AM, 
      the values in cells $[0, t_2-1]$ are not altered. 
      The value $\maybe$ which was set in $Q^{i-1-t_2}[0]$ 
      simply shifts through $Q^{l''}[l''-(i-1-t_2)]$ unaltered.
\item At time $l$, we know that $x^{l} \entails \beta$.
	As per the definition of the AM, the value in cells $Q^{l}[t_1, t_2]$ is set to $\true$ if the current value is $\maybe$.
		This range includes the cell $Q^{l}[l-(i-1-t_2)]$, whose old value was $\maybe$.
\item for all times $l''' \in [l+1, i-1]$, as per the definition of the AM, the value
	$\true$ which was set in $Q^{l}[l-(i-1-t_2)]$ simply shifts through $Q^{l'''}[l'''-(i-1-t_2)]$ unaltered.
	Thus, at time $i-1$, the value in cell $Q^{i-1}[t_2]$ is $\true$.

\end{itemize}

	In the next time step, when
	$\maybe$ is inserted, the verdict $r^{i-1-t_2}$ moves to $Q^{i}[t_2+1]$ as required,
	as the AM does not modify cell $t_2+1$.\\

(\emph{Assume $x^{i-t_2-1} \doesnotentail \alpha \until_{[t_1,t_2]} \beta$}):

There are three scenarios where this can happen. We consider each scenario independently.\\

\emph{Scenario 1:}
Let there be some $k$ in $[i-1-t_2,i-1-t_2+t_1-1]$ such that $x^k \doesnotentail \alpha$.
We consider this interval as follows:
\begin{itemize}
\item At time $i-1-t_2$, as per the definition of the AM, a value of $\maybe$
		is inserted into the queue at index $0$. This value corresponds to the verdict $r^{i-1-t_2}$.
\item For all times $l' \in [i-1-t_2, k-1]$, we know that 
      $x^{l'} \entails \alpha$ and as per the definition of the AM, 
      the values in cells $[0, t_1-1]$ are not altered. 
      The value $\maybe$ which was set in $Q^{i-1-t_2}[0]$ 
      simply shifts through $Q^{l'}[l'-(i-1-t_2)]$ unaltered.
\item At time $k$, we know that $x^{k} \doesnotentail \alpha$.
	As per the definition of the AM, all cells in $[0,t_1-1]$
	are set to $\false$. This includes the cell $Q^k[k-(i-1-t_2)]$, whose old value was $\maybe$.
\item for all times $l'' \in [k+1, i-1]$, as per the definition of the AM, the value
	$\false$ which was set in $Q^{k}[k-(i-1-t_2)]$ simply shifts through $Q^{l''}[l''-(i-1-t_2)]$ unaltered.
	Thus, at time $i-1$, the value in cell $Q^{i-1}[t_2]$ is $\false$.
\end{itemize}
	In the next time step, when
	$\maybe$ is inserted, the verdict $r^{i-1-t_2}$ moves to $Q^{i}[t_2+1]$ as required,
	as the AM does not modify cell $t_2+1$.\\

\emph{Scenario 2:}
Let $l$ be the earliest time instance in $[i-1-t_2+t_1,i-1]$ such that $x^{l} \entails \beta$ and 
for all $k'$ in $[i-1-t_2,i-1-t_2+t_1-1],~x^{k'} \entails \alpha$ and
there exists some $k''$ in $[i-1-t_2+t_1,l-1]$ such that $x^{k''} \doesnotentail \alpha$.
We break the interval $[i-1-t_2,i-1]$ into five parts:
\begin{itemize}
\item At time $i-1-t_2$, as per the definition of the AM, a value of $\maybe$
		is inserted into the queue at index $0$, and this value of $Q^{i-1-t_2}[0]$
		is not altered in this time unit. This value corresponds to the verdict $r^{i-1-t_2}$.
\item For all times $l' \in [i-1-t_2+1, i-1-t_2+t_1-1]$, we know that 
      $x^{l'} \entails \alpha$ and as per the definition of the AM, 
      the values in cells $[0, t_1-1]$ are not altered. 
      The value $\maybe$ which was set in $Q^{i-1-t_2}[0]$ 
      simply shifts through $Q^{l'}[l'-(i-1-t_2)]$ unaltered.
\item For all times $l'' \in [i-1-t_2+t_1, k''-1]$, we know that 
	$x^{l''} \entails \alpha$ and $x^{l''} \doesnotentail \beta$.
	As per the definition of the AM, 
      the values in cells $[0, t_2-1]$ are not altered. 
      The value $\maybe$ which was set in $Q^{i-1-t_2}[0]$ 
      simply shifts through $Q^{l''}[l''-(i-1-t_2)]$ unaltered.
\item At time $k''$, we know that $x^{k''} \doesnotentail \alpha$ and $x^{k''} \doesnotentail \beta$.
	As per the definition of the AM, the value in cells $Q^{k''}[t_1, t_2]$ is set to $\false$ if the current value is $\maybe$.
		This range includes the cell $Q^{k''}[k''-(i-1-t_2)]$, whose old value was $\maybe$.
\item for all times $l''' \in [k''+1, i-1]$, as per the definition of the AM, the value
	$\false$ which was set in $Q^{k''}[k''-(i-1-t_2)]$ simply shifts through $Q^{l'''}[l'''-(i-1-t_2)]$ unaltered.
	Thus, at time $i-1$, the value in cell $Q^{i-1}[t_2]$ is $\false$.
\end{itemize}
	In the next time step, when
	$\maybe$ is inserted, the verdict $r^{i-1-t_2}$ moves to $Q^{i}[t_2+1]$ as required,
	as the AM does not modify cell $t_2+1$.\\

\emph{Scenario 3:}
For all $l$ in $[i-1-t_2+t_1,i-1]$, $x^{l} \doesnotentail \beta$ and 
for all $k$ in $[i-1-t_2, i-1$], $x^{k} \entails \alpha$.
We break the interval $[i-1-t_2,i-1]$ into five parts:
\begin{itemize}
\item At time $i-1-t_2$, as per the definition of the AM, a value of $\maybe$
		is inserted into the queue at index $0$, and this value of $Q^{i-1-t_2}[0]$
		is not altered in this time unit. This value corresponds to the verdict $r^{i-1-t_2}$.
\item For all times $l' \in [i-1-t_2+1, i-1-t_2+t_1-1]$, we know that 
      $x^{l'} \entails \alpha$ and as per the definition of the AM, 
      the values in cells $[0, t_1-1]$ are not altered. 
      The value $\maybe$ which was set in $Q^{i-1-t_2}[0]$ 
      simply shifts through $Q^{l'}[l'-(i-1-t_2)]$ unaltered.
\item For all times $l'' \in [i-1-t_2+t_1, i-2]$, we know that 
	$x^{l''} \entails \alpha$ and $x^{l''} \doesnotentail \beta$.
	As per the definition of the AM, 
      the values in cells $[0, t_2-1]$ are not altered. 
      The value $\maybe$ which was set in $Q^{i-1-t_2}[0]$ 
      simply shifts through $Q^{l''}[l''-(i-1-t_2)]$ unaltered.
\item At time $i-1$, we know that
	$x^{i-1} \entails \alpha$ and $x^{i-1} \doesnotentail \beta$.
	As per the definition of the AM, the value in cells $Q^{i-1}[t_2]$ is set to $\false$ as the current value is $\maybe$.
\end{itemize}
	In the next time step, when
	$\maybe$ is inserted, the verdict $r^{i-1-t_2}$ moves to $Q^{i}[t_2+1]$ as required,
	as the AM does not modify cell $t_2+1$.\\

	Theorem~\ref{lab-correctness-thm} for $\phi=\alpha \until_{[0,t_2]} \beta$
	can be proven in a similar manner.

\end{proof}

\subsubsection*{(Proof of Theorem~\ref{lab-correctness-thm2})}:
\begin{proof}
This is because of construction of $EM$ for all single operator MTL formulae $\phi$.
	As per Theorem~\ref{lab-correctness-thm}, the verdict corresponding to time $j-l(\odot)$
	is available at time $j$, and this verdict is available in cell $Q^j[l(\odot)]$.
An AM does not alter the contents of cells at positions
	greater than $l(\odot)$. Therefore, at time $j+k$, due to $k$ shifts,
	the verdict corresponding to time $j-l(\odot)$ is available unaltered in cell $Q^{(j+k)}[l(\odot) + k]$. 
\end{proof}

\subsubsection*{(Proof of Proposition~\ref{prop:simul_modify})}:
\begin{proof}
	From Table~\ref{tab:em_operations}, we can see that machine $M_1$ 
        modifies $I_1=[0,t_1-1]$ 
	and
						       $M_2$ modifies $I_2=[t_1,t_2]$ and $I_3=[t_2,t_2]$ and 
						       $M_3$ modifies $I_4=[t_1,t_2-1]$. 

	The pairs of overlapping intervals are $(\mathcal{I}_2,\mathcal{I}_3)$ and $(\mathcal{I}_2,\mathcal{I}_4)$.
	The interval $\mathcal{I}_2$ and $\mathcal{I}_3$ are modified by the same machine $M_2$, for different values of 
	$\alpha_1$, which is not possible in a given cycle, so the simultaenous modification of both intervals is not possible. 

	The interval $\mathcal{I}_2$ is modified by $M_2$ when $\alpha_1=\true$. 
	The interval $\mathcal{I}_4$ is modified by $M_3$ when $\alpha_1 \lor \alpha_0 =\false$, implying $\alpha_1=\false$.         Thus, the simultaenous modification of both intervals is not possible. 

\end{proof}

\end{document}